\documentclass[11pt]{article}
\usepackage[margin=1in]{geometry}
\usepackage{mdframed}
\usepackage{times} %
\usepackage{amsmath,amssymb,mathtools,bm,amsthm}
\usepackage{graphicx}
\usepackage{xcolor}
\usepackage{bbm}
\usepackage[backref=page,breaklinks=true,hidelinks]{hyperref}
\hypersetup{
    colorlinks,
    linkcolor={red!50!black},
    citecolor={blue!50!black},
    urlcolor={blue!80!black}
}
\renewcommand*{\backref}[1]{}
\renewcommand*{\backrefalt}[4]{%
    \ifcase #1 {\footnotesize(Not cited.)}%
    \or        {\footnotesize(Cited on page~#2.)}%
    \else      {\footnotesize(Cited on pages~#2.)}%
    \fi}
    
\usepackage{natbib}
\usepackage{enumitem}
\usepackage{booktabs}
\usepackage{subcaption}
\usepackage[ruled]{algorithm2e}
\SetKwInput{KwInput}{Input}
\SetKwInput{KwOutput}{Output}

\SetAlFnt{\small}
\SetAlCapFnt{\small}
\SetAlCapNameFnt{\small}
\SetAlCapHSkip{0pt}
\IncMargin{-\parindent}

\RequirePackage{bm} 
{%
\begin{proof}[#1]  %
}%
{%
\end{proof}  %
}

\newtheorem{theorem}{Theorem} 

\newtheorem{lemma}{Lemma}
\newtheorem{proposition}{Proposition}

\newtheorem{remark}{Remark}

\newtheorem{assumption}{Assumption}[section]

\newtheorem{example}{Example}

\renewcommand{\hat}{\widehat}
\renewcommand{\tilde}{\widetilde}

\def\min{\qopname\relax n{min}}
\def\max{\qopname\relax n{max}}

\def\argmin{\qopname\relax n{argmin}}
\def\argmax{\qopname\relax n{argmax}}

\def\supp{\qopname\relax n{\mathbf{supp}}}

\newenvironment{lp*}{\begin{equation*}  \begin{array}{lll}}{\end{array}\end{equation*}}

\usepackage[colorinlistoftodos]{todonotes}
\presetkeys{todonotes}{inline}{}
\presetkeys{todonotes}{color=orange!10}{}

\renewcommand{\P}{\mathbb{P}}
\newcommand{\R}{\mathbb{R}}
\newcommand{\Var}{\operatorname{Var}}

\newcommand{\E}{\mathbb{E}}
\newcommand{\Ind}{\mathbbm{1}}
\newcommand{\de}{\textnormal{d}}

\makeatletter
\long\def\blfootnote#1{%
  \begingroup
    \renewcommand\thefootnote{}%
    \renewcommand\@makefntext[1]{\noindent ##1}%
    \footnotetext{#1}%
  \endgroup}
\makeatother

\title{Algorithmic Feature Highlighting \\
for Human--AI Decision-Making
}
\author{Yifan Guo \qquad Jann Spiess   \\ { Stanford University}}

\date{April 23, 2026}

\begin{document}

\maketitle

\begin{abstract}
Human decision-makers often face choices about complex cases with many potentially relevant features, but limited bandwidth to inspect and integrate all available information. In such settings, we study algorithms that \emph{highlight} a small subset of case-specific features for human consideration, rather than producing a single prediction or recommendation. We model highlighting as a constrained information policy that selects a small number of features to reveal. A central issue is how humans interpret the algorithm's choice of features: a \emph{sophisticated} agent correctly conditions on the selection rule, while a \emph{naive} agent updates only on revealed feature values and treats the selection event as exogenous. We show that optimizing highlighting for sophisticated agents can be computationally intractable, even in simple discrete and binary settings, whereas optimizing for naive agents is tractable as long as the maximal bandwidth is fixed. We also show that a highlighting policy that is optimal for sophisticated agents can perform arbitrarily poorly when deployed to naive agents, motivating robust, implementable alternatives. We illustrate our framework in a calibrated empirical exercise based on the American Housing Survey. Overall, our results establish the value of highlighting a context-specific set of features rather than a fixed one as a practically appealing and computationally feasible tool for achieving human--algorithm complementarity.
\end{abstract}

\blfootnote{
    Yifan Guo (\href{mailto:yifanguo@stanford.edu}{yifanguo@stanford.edu}) and Jann Spiess (\href{mailto:jspiess@stanford.edu}{jspiess@stanford.edu}), Graduate School of Business, Stanford University.
We thank Haozhan Gao and three anonymous referees for helpful comments.
}

\section{Introduction}

In high-stakes decision-making, algorithms frequently assist rather than replace human decision-makers.
For example, algorithms may screen resumes, but final hiring decisions are still made by managers.
In such settings, having algorithms recommend a final action may not be sufficient for decision-makers to make high-quality decisions that reflect their private information and preferences.
Instead, effective algorithmic support tools should surface interpretable and verifiable details that are most relevant to the final decision.

In this article, we propose a model of joint human--algorithm decision-making in which the algorithm \emph{highlights} key features that are relevant to the downstream decision of the human decision-maker. We model highlighting as a constrained information policy that selects which of a large number of features to reveal. We show that simple and computationally tractable algorithms are able to improve markedly over policies that instead always reveal the same features, and that this gain is robust to different ways in which decision-makers may interpret the selected features. Overall, we argue that highlighting presents a practically appealing and computationally feasible tool for achieving human--algorithm complementarity that raises rich theoretical questions and warrants further investigation.

We model human--algorithm interaction as a principal--agent game between an algorithm designer (principal) who decides which features to highlight and a human decision-maker (agent) who observes the highlighted features and makes a decision.
The joint goal of principal and agent is to choose an action for a decision instance that is characterized by a potentially high-dimensional feature vector.
The algorithm designer specifies an algorithm that observes the full feature vector and selects a subset of its components to report to a bandwidth-constrained decision-maker.
The human decision-maker observes only those features highlighted by the algorithm.
After receiving the highlighted features, the agent infers missing information and then chooses an action.
For example, when a candidate applies for a job, the algorithm may observe all information contained in the application and then strategically highlight a subset of features for the human HR manager, who makes the final decision.
Throughout, we assume that preferences between the algorithm designer and the decision-maker are aligned.

We believe that our focus on \emph{highlighting} features rather than providing predictions or summaries has two important advantages in practice:
First, presenting more than a simple prediction or recommended action allows the human decision-maker to learn richer information that provides important context and can be combined with private information.
Second, the constraint that the algorithm can only pick existing features means that the provided information is interpretable, trustworthy, and verifiable.
For example, simply providing a summary generated by large-language models (LLMs) could be problematic since we cannot generally guarantee that it contains only factual information about a given case.
By contrast, even if we use an LLM to decide \emph{which} features to highlight, the resulting information remains correct by construction, even without strong guarantees about the underlying algorithm.%
\footnote{This would still leave space for deception by a misaligned principal since the choice of \emph{which} features are selected could be driven by a different objective \citep[akin to Bayesian persuasion,][]{kamenica2011bayesian}, but the deception would then mainly be driven by failing to highlight the most relevant information rather than misrepresenting feature values.}

We consider different modes of human decision-making to highlight potential limitations in practice. 
We focus on two benchmark models of human decision-makers: a sophisticated agent and a naive agent.
Both agents observe the features highlighted by the algorithm and update their posteriors accordingly.
The \emph{sophisticated} agent also takes into account the information contained in the algorithm's decision about which features to highlight, in addition to the value of the highlighted features.
By contrast, the \emph{naive} agent ignores why the algorithm highlights specific features and processes only the highlighted features at face value.
For example, a hiring manager could be sophisticated if they respond not just to the information contained in a summary, but also to what is left out. Or they could be naive if they simply fill in the blanks based on their best guess, without taking into account why some information may be included or left out.
In practice, a real decision-maker may behave as a mixture of these two types, exhibiting both sophisticated and naive behaviors.

We show that these different modes of updating have vastly different implications for optimal highlighting policies, their robustness, and their computation.
We first document gaps between the optimal policies for naive and sophisticated agents, and the cost of optimizing for one when the audience is the other.
While the sophisticated agent always does at least as well as the naive agent, we show in simple examples that deploying a policy that targets the naive agent may still lead to an arbitrarily large gap from the \emph{optimal} policy for a sophisticated audience (``price of simplicity''). 
Conversely, we demonstrate that a highlighting policy optimized for a sophisticated agent can have arbitrarily bad performance for a naive agent (``price of complexity'') relative to the latter's optimal policy.
Beyond the potential cost to naive agents, we illustrate that optimal sophisticated solutions tend to be overly complex.

Optimal solutions that optimize for a sophisticated decision-maker not only lack robustness; they are also often infeasible to compute.
Since the sophisticated agent interprets not only the revealed features, but also takes into account \emph{why} they were revealed, the principal's optimization problem should account for this additional information channel.
As with hardness results in algorithmic Bayesian persuasion, we show that the optimal highlighting problem for a sophisticated agent is NP-hard.
Therefore, aiming for first-best highlighting schemes for sophisticated decision-makers may lead to overly complex and unintuitive solutions that are both computationally intractable and lack robustness against suboptimal agent decisions.

In contrast to the sophisticated case, even first-best optimal solutions for naive updating can generally be computed in polynomial time. Beyond such first-best solutions, we also consider even simpler greedy solutions that select features based on the incremental information they provide. Focusing on the case of independent binary features, we show that this greedy solution yields first-best outcomes for the naive agent and provides additional gains for the sophisticated agent. Specifically, we provide an analytical characterization of the asymptotic average loss for both types of agents, which we extend to the case where there may be limited correlation across features.

We show that the simple algorithms we consider can substantially outperform a benchmark in which a fixed set of features is highlighted. 
By highlighting ex-post surprises, we can generally obtain considerably smaller losses across both sophisticated and naive updaters.
We show that these algorithms have some surprising properties.
For example, in our setting, it is not generally optimal to always highlight the maximal allowable number of features, since additional information does not always lead to better downstream decisions.
Instead, optimal algorithms may withhold incremental information if that information misleads more than it helps.

We illustrate these findings in a calibrated numerical exercise using data from the American Housing Survey. Using data from around 1,500 homes from the Minneapolis--St.~Paul metro area and a Gaussian belief model over 44 housing characteristics, we show that contextual highlighting policies improve substantially over fixed policies: for example, when highlighting up to $k=10$ features, a contextual greedy algorithm cuts loss by more than half. The exercise also confirms the non-monotonicity result in practice: a contextual greedy policy with ten features outperforms even full revelation of all 43 non-price features.
Overall, the exercise demonstrates how a relatively simple, computationally efficient greedy algorithm can produce highly informative posteriors even when updating is naive.

The assumptions of perfect alignment, joint priors, and (partially) rational updating represent substantive limitations of our main model, which suggest important further work.
We therefore also discuss extensions to more realistic updating rules and additional frictions between principal and agent.
Beyond sophisticated and naive agents, we consider specific behavioral biases in which agents update their beliefs based on stereotypes.
We also discuss private information of the agent, wrong beliefs, and misaligned preferences.
In addition to these extensions, we discuss more complex feature distributions.
With Normally distributed features, we show that sophisticated solutions take unintuitive and complex forms, reinforcing our prescription to optimize for naive updating.

The remaining article is structured as follows.
\autoref{sec:literature} reviews the related literature in human--AI interaction, information design, and feature revelation.
\autoref{sec:model} introduces the model of algorithmic feature highlighting with naive and sophisticated updating.
\autoref{sec:updating} characterizes the cost of misoptimization when the designer targets the wrong updating type.
\autoref{sec:complexity} studies computational complexity, establishing hardness for sophisticated updating and tractability for naive updating, and analyzes greedy algorithms.
\autoref{sec:binary} provides closed-form optimal policies and asymptotic loss characterizations for independent binary features.
\autoref{sec:empirical} presents our simulation study based on the American Housing Survey.
\autoref{sec:extensions} discusses extensions, and \autoref{sec:conclu} concludes.

\section{Related Literature}
\label{sec:literature}

We contribute to an interdisciplinary literature that aims to model and develop solutions for human--algorithm interaction. We also build upon approaches from information design and the literature on subset selection under constraints.

\subsection{Human--AI Decision Making and Behavioral Foundations}

We relate to a growing literature that studies how humans and algorithms jointly make decisions.
\citet{noti2025ai} formalizes AI-assisted decision making with a learning human and shows that the optimal feature-revelation schedule balances short-run accuracy against long-run human learning, yielding a tractable combinatorial characterization.
\citet{greenwood2025designing} analyzes instance-contingent algorithmic delegation; it shows that the optimal delegate can be strictly better than the best standalone AI and provides both hardness results and efficient algorithms.
\citet{tan2026should} studies the complementary question of \emph{when} to dispatch a human reviewer under a budget on human effort, deriving a closed-form optimal rule and applying it to AI-assisted peer review.
\citet{athey2020allocation} studies the allocation of decision authority between humans and AI.
\citet{boyaci2024human} studies how machine input affects accuracy and effort, while \citet{green2019principles} articulates practical principles for algorithm-in-the-loop decision making.
\citet{balakrishnan2025human} shows that feature transparency reduces naive advice-weighting bias in human--algorithm collaboration with private information.
More broadly, work on designing human--AI collaboration under behavioral responses \citep[e.g.,][]{agarwal2025designing,mclaughlin2022algorithmic} studies how optimal disclosure, recommendation, and automation designs change when humans exhibit biased responses.
\citet{bastani2026improving} studies sequential decision-making and proposes an RL-based approach to produce simple, interpretable ``tips'' targeted to bottlenecks in the human policy.
\citet{mclaughlin2024designing,hoong2025improving} provide frameworks for algorithmic recommendations that allow for suboptimal responses.
Our framework differs in that the AI never acts or makes recommendations directly; instead, it strategically selects which features to surface,
and we characterize how performance depends on whether the human anticipates (or neglects) the policy dependence of revealed information.

A closely related strand uses decision theory to characterize the value of information that one agent can contribute to another's decision.
\citet{guo2025value} formalizes human--AI complementarity through the expected improvement in best-attainable payoff from an additional signal, and \citet{guo2026complllm} operationalizes this idea by fine-tuning an LLM to extract signals from a supervisor's context that complement an existing agent's recommendation.
These works take the space of signals as exogenously given and assume that the receiver best-responds to the induced posterior.
By contrast, our signals arise endogenously as a subset of the underlying feature vector, which ensures that the revealed content is interpretable and verifiable by construction.
\citet{vaccaro2024when} conducts a meta-analysis of 106 studies, finding that human--AI combinations on average perform worse than the best of humans or AI alone.
\citet{angelova2023algorithmic} documents that many judges underperform when overriding algorithmic recommendations.
 \citet{alur2024human} proposes a framework for incorporating human expertise by distinguishing algorithmically indistinguishable inputs.
\citet{donahue2022human} provides theoretical conditions under which complementarity is achievable or impossible.

\subsection{Algorithmic Information Design and Persuasion}

Our revelation policy can be viewed as a signaling scheme, directly connecting our setting to the literature on information design and Bayesian persuasion.
The foundational work of \citet{kamenica2011bayesian} introduced Bayesian persuasion, formalizing how a sender commits to an information structure to influence a receiver's action; a closely related precursor, \citet{rayo2010optimal}, characterizes optimal disclosure by a sender who commits ex ante to a signal structure and shows that optimal disclosure generically pools states into coarse partitions.
This framework has since been extended across economics and operations research: \citet{candogan2020information} surveys applications in operations, and \citet{dughmi2016algorithmic} develops algorithmic tools for computing persuasion schemes with complexity guarantees.

We adopt this information-design perspective but introduce behavioral heterogeneity by modeling the receiver as either naive (ignoring the selection event) or sophisticated (conditioning on the revelation rule), and we derive structural and computational results for optimal revelation under both types.
Relatedly, \citet{yu2023encoding} combines information design with behavioral modeling via deep learning, demonstrating how to incorporate empirically estimated human responses into signaling schemes.
In contrast, we focus on an analytically tractable model of belief formation (naive vs.\ sophisticated) and provide explicit guarantees.
The closest work to ours is \citet{hoefer2024algorithmic}, which studies algorithmic persuasion with constraints on the design space imposed by verifiable evidence.
Their model allows the sender to choose among evidentiary messages subject to a feasibility constraint, whereas in our setting, the highlighted information must correspond to a subset of the state's features.
Related work on privately informed receivers, including \citet{guo2019interval}, characterizes the sender-optimal disclosure policy as an interval structure, providing a tight structural benchmark for persuasion when the receiver holds private information.

\subsection{Subset Selection and Instance-Wise Feature Revelation}
\label{sec:subsets}

Our model is closely related to the broad literature on subset selection under a cardinality constraint.
In our setting, the most direct benchmark is fixed subset selection: the algorithm commits ex ante to a single set of features to highlight and uses the same set for every instance.
This benchmark parallels classical feature selection and sparse approximation problems, where the goal is to select a small number of coordinates that best predict an outcome or reconstruct the remaining coordinates.
Such problems are often computationally challenging in the worst case, motivating approximation algorithms and structural conditions under which greedy methods are effective.
In the context of selecting predictive variables, \citet{das2011submodular} analyzes greedy subset selection for regression-type objectives via spectral conditions and the submodularity ratio, providing guarantees even when exact submodularity does not hold.
Relatedly, observation and feature-acquisition objectives in probabilistic models are often approximately submodular, yielding principled greedy policies \citep{krause2007near}.
However, our setting focuses on instance-wise feature selection, which can strictly outperform the fixed-subset benchmark.

\section{A Model of Algorithmic Feature Highlighting}\label{sec:model}

We model the interaction between an algorithm that highlights features and a bandwidth-constrained human decision-maker who updates their belief and takes actions based on the highlighted features.
We assume that a decision instance is given by a random vector $X = (X_i)_{i \in M} \sim \P$ with a finite index set $M$ of size $|M| = d$.%
\footnote{In \autoref{sec:extensions} we discuss extensions to human private information, in which case features are given by $X = (X_M,X_H)$ with machine features $X_M$ and human features $X_H$.}

\paragraph{Highlighting policy.}
We model feature \emph{highlighting} as follows. 
We assume that the features $X = (X_i)_{i \in M} \in \mathcal{X}$ are accessible to an algorithm
that selects (or ``highlights'') a subset $I \subseteq M$ with $|I| \leq k$ and reveals the set $I$ as well as the associated features $X_I = (X_i)_{i \in I}$ to the human decision-maker. Here, $k$ is an exogenously given bandwidth constraint that limits the number of features that can be highlighted.
A (deterministic) highlighting policy is a map
$$
\sigma: \mathcal{X} \to \{I \subseteq M: |I|\le k\}.
$$
Randomized policies can be modeled similarly; deterministic policies suffice for our main hardness and separation results.
We assume that this highlighting policy is chosen by an \emph{algorithm designer}, who acts as the principal of our game.

Here, the highlighting policy can be interpreted as follows. 
In hiring, a company may use an algorithm to assist HR in the recruiting process; the algorithm highlights several key features in each applicant's resume. 
In a hospital, the algorithm may highlight key components of a patient's medical file to assist doctors in their final assessment. 
In school admissions, the algorithm may highlight key aspects of each candidate's background for the admissions committee.

\paragraph{Human belief formation.}
Upon observing $(I, X_I)$, the human decision-maker forms a posterior belief over the unknown components of $X$.
We assume that the human decision-maker knows the distribution $\P$ (which acts as their prior), and consider two main ways in which the decision-maker updates to form a posterior:

\begin{itemize}
    \item \emph{Sophisticated updating.}
    The sophisticated agent is aware of the prior $\P$ and the policy $\sigma$, and hence conditions on the selection event.
    We model this by the distribution
    $$
    \hat{\P}^{S}(\cdot \mid I, X_I)
    =
    \P\bigl(\cdot \mid X_I, \sigma(X)=I\bigr).
    $$
    Equivalently, $\hat{\P}^{S}$ is the Bayesian posterior under the joint distribution induced by $(\P,\sigma)$.

    \item \emph{Naive updating.}
    The naive agent updates on the revealed values $X_I$, but ignores the informational content of \emph{which} set $I$ was selected.
    We model this by the distribution
    $$
    \hat{\P}^{N}(\cdot \mid I, X_I)
    =
    \P(\cdot \mid X_I),
    $$
    i.e., the naive agent treats the selection of $I$ as exogenous. Note that this update does not require knowledge of $\sigma$.
\end{itemize}

Here, we assume for simplicity that updating is based on the \emph{true} distribution $\P$. However, our model directly extends to the case where updating is performed according to some prior $\tilde{\P}$ that may or may not be equal to $\P$. In that case, we can still write $\hat{\P}^S$ and $\hat{\P}^N$ for the corresponding posteriors.

\paragraph{Decision and loss.}
After forming a belief $\hat{\P}(\cdot \mid I, X_I)$, the human chooses an action $a \in A$ to minimize expected loss.
Given a loss function $\ell: A \times \mathcal{X} \to \R$, the human plays a Bayes action
$$
\hat a(\ell, \hat{\P}) \in \argmin_{a\in A} \E_{\hat{\P}}[\ell(a,X)].
$$
The loss at instance $X$ is then
$
L(\hat{\P}, X) \;=\; \ell\!\bigl(\hat a(\ell,\hat{\P}), X\bigr).
$

\paragraph{Design objective.}
We focus on aligned objectives: the algorithm designer seeks to minimize expected loss according to the same loss function.
For agent type $\tau\in\{S,N\}$ and policy $\sigma$, define the expected loss (risk) as
$$
R^\tau(\sigma)
~:=~
\E\Bigl[
L\bigl(\hat{\P}^{\tau}(\cdot \mid I, X_I),\, X\bigr)
\Bigr],
\qquad
I=\sigma(X).
$$
The design problem is
$$
\min_{\sigma:\,|\,\sigma(x)\,|\le k} \; R^\tau(\sigma).
$$
If the designer faces uncertainty over human types, a natural Bayesian objective is a weighted average
$\lambda R^S(\sigma)+(1-\lambda)R^N(\sigma)$ for some population share $\lambda\in[0,1]$.
A robust alternative is to minimize the worst-case loss across types, $\min_{\sigma:\,|\,\sigma(x)\,|\le k} \max\{R^S(\sigma),R^N(\sigma)\}$.
In settings where sophisticated updating weakly improves performance relative to naive updating for every policy, this minimax problem reduces to minimizing naive risk.

\paragraph{Special case: recovery.}
A common special case is recovery of a target quantity $y=y(x)\in\R^m$.
Let $A=\R^m$ and define squared loss
$$
\ell(a,x) = \|a-y(x)\|_2^2.
$$
Then $\hat a(\ell,\hat{\P})=\E_{\hat{\P}}[y(X)]$ is the posterior mean.

\paragraph{Robust microfoundation for squared loss.}
The squared loss above admits a robust interpretation when the downstream decision problem is not known to the algorithm designer.
Suppose the human's action $a\in\R^m$ feeds into a downstream task whose payoff depends on an unknown linear combination $\beta^\top a$ of the reported action, evaluated against the target $\beta^\top y(x)$ under quadratic loss
$(\beta^\top(a-y(x)))^2$.
If the coefficient $\beta\in\R^m$ is unknown and may be chosen adversarially from the unit ball $\{\beta:\|\beta\|_2\le 1\}$, the worst-case loss reduces to
$$
\sup_{\|\beta\|_2\le 1}\;\bigl(\beta^\top(a-y(x))\bigr)^2 \;=\; \|a-y(x)\|_2^2,
$$
which is exactly the squared loss above.

\paragraph{Special case: outcome-targeted recovery.}
In our simulations, we consider a specific loss function that combines full recovery of the feature vector with prediction loss for a specific target.
In this case, $y(x) = \begin{psmallmatrix} x \\ \E[Y \mid X=x] \end{psmallmatrix}$ for some target outcome $Y$, and
$$
\ell(a, x) = (1-\alpha) \, \|\hat{x} - x\|_2^2 + \alpha \, (\hat{y} - \E[Y \mid X=x])^2
$$
for $a = \begin{psmallmatrix} \hat{x} \\ \hat{y} \end{psmallmatrix}$.
Full recovery ($\alpha = 0$) as well as simple prediction ($\alpha = 1$) are special cases.

\section{Naive vs.\ Sophisticated Updating and the Cost of Misoptimization}\label{sec:updating}

Above, we have introduced a model of feature highlighting in which an algorithm sends a subset of features to a human decision-maker.
There, we have considered two different ways in which the human decision-maker may update their belief: either by taking the selection rule into account (sophisticated), or by ignoring selection (naive). We now show that these two ways of updating lead to different solutions, and that optimizing for one type of agent comes at a price when the policy is deployed for the other one.

Let $\sigma^N \in \argmin_\sigma R^N(\sigma)$ and $\sigma^S \in \argmin_\sigma R^S(\sigma)$ denote optimal policies for naive and sophisticated agents, respectively.
To quantify the mismatch when a policy optimized for one agent type is deployed to the other, define
\begin{align}\label{eq:gap-def}
&\text{price of complexity: }
&
\Delta^{N} &\;:=\; R^{N}(\sigma^{S}) - R^{N}(\sigma^{N}) \;\ge 0;
\\
&\text{price of simplicity: }
&
\Delta^{S} &\;:=\; R^{S}(\sigma^{N}) - R^{S}(\sigma^{S}) \;\ge 0.
\end{align}
The next two examples show that these gaps can be arbitrarily large in either direction: a policy that is optimal for sophisticated agents may perform extremely poorly when deployed to naive agents (large $\Delta^N$), and conversely a policy optimized for naive agents can be arbitrarily suboptimal for sophisticated agents (large $\Delta^S$).

\begin{example}[A sophisticated--optimal scheme can mislead a naive agent]\label{ex:parity}
Let $X=(X_1,X_2)$ where $X_1\in\{0,1\}$ and $X_2\in\{0,B\}$, and suppose $X$ is uniform over the four states.
Assume $B$ is odd.
The machine observes $X$ and reveals at most one coordinate ($k=1$).
Consider squared recovery loss with $y(x)=x$, i.e., $\ell(a,x)=\|a-x\|_2^2$.
Define a deterministic highlighting policy $\sigma^{S}$ by
$$
\sigma^{S}(x)=
\begin{cases}
\{1\}, & \text{if } x_1 + x_2 \text{ is odd},\\
\{2\}, & \text{otherwise}.
\end{cases}
$$
A sophisticated agent can recover $X$ with zero loss: from $(I,x_I)$ the agent infers the unrevealed coordinate using the fact that $\sigma^{S}(X)=I$.
Hence $R^S(\sigma^{S})=0$.

In contrast, a naive agent ignores the informational content of $I$.
When $I=\{1\}$, the naive posterior predicts $X_2$ by its conditional mean $\E[X_2\mid X_1]=\frac{B}{2}$, incurring squared error $(\frac{B}{2})^2=\frac{B^2}{4}$.
When $I=\{2\}$, the naive posterior predicts $\E[X_1\mid X_2]=1/2$, incurring squared error $1/4$.
Each event occurs with probability $1/2$, so
$
R^N(\sigma^{S})=\frac12\cdot\frac{B^2}{4}+\frac12\cdot\frac14=\frac{B^2}{8}+\frac18.
$
The naive--optimal policy is to always reveal $X_2$, i.e., $\sigma^N(x)\equiv\{2\}$, which yields
$
R^N(\sigma^N)=\E[(X_1-\E[X_1])^2]=1/4.
$
Therefore
$$
\Delta^N \;=\; R^N(\sigma^{S})-R^N(\sigma^{N})
=
\frac{B^2-1}{8},
$$
which can be arbitrarily bad relative to the optimal loss of $R^N(\sigma^{N}) = 1/4$.
\end{example}

\begin{example}[A naive--optimal scheme can be arbitrarily suboptimal for a sophisticated agent]\label{ex:parity2}
Let $X=(X_1,X_2)$ where $X_1\in\{0,B\}$ and $X_2\in\{0,B\}$, and suppose $X$ is uniform over the four states.
The machine observes $X$ and reveals at most one coordinate ($k=1$).
Consider squared recovery loss $\ell(a,x)=\|a-x\|_2^2$.
Define a deterministic policy $\sigma^{S}$ by
$$
\sigma^{S}(x)=
\begin{cases}
\{1\}, & \text{if } x_1=x_2,\\
\{2\}, & \text{otherwise}.
\end{cases}
$$
A sophisticated agent can recover $X$ with zero loss: if $I=\{1\}$ then $x_2=x_1$, while if $I=\{2\}$ then $x_1=B-x_2$.
Hence $R^S(\sigma^S)=0$.

Now consider a naive--optimal policy. By symmetry, a naive agent achieves the same expected loss by always revealing either coordinate.
Take $\sigma^N(x)\equiv\{1\}$.
Since the selection is constant, sophistication provides no additional inference under $\sigma^N$, so $R^S(\sigma^N)=R^N(\sigma^N)$.
Under $\sigma^N$, the posterior predicts $X_2$ by its mean $\E[X_2]=\frac{B}{2}$, yielding
$
R^S(\sigma^N)=R^N(\sigma^N)=\E[(X_2-\frac{B}{2})^2]=\frac{B^2}{4}.
$
Therefore
$$
\Delta^S
=
R^S(\sigma^N)-R^S(\sigma^S)
=
\frac{B^2}{4},
$$
which can be arbitrarily bad relative to the optimal loss of $R^S(\sigma^{S}) = 0$.
\end{example}

Together, \autoref{ex:parity} and \autoref{ex:parity2} show that optimizing a highlighting policy for one belief-formation model can lead to arbitrarily poor performance when deployed to the other model.
Furthermore, optimal sophisticated policies may take unintuitive forms that rely on complex relationships across the full feature vector.
In \autoref{sec:continuous}, we show that these findings carry over to the case where the features $X$ are continuous: even when they are Normally distributed and relatively low-dimensional, sophisticated solutions take unintuitive and hard-to-interpret forms.
These findings motivate a search for simple, robust policies that remain competitive across agent types.

\section{Computational Complexity}\label{sec:complexity}

Above, we have shown that optimal highlighting policies can differ substantively by whether they optimize for naive or sophisticated decision-makers.
We now show that this difference extends to our ability to compute optimal policies.

To show a computational contrast between sophisticated and naive receivers, we restrict our attention to discrete instance distributions.
Let $n:=|\supp(\P)|$ denote the number of possible instances in the support, $d$ the number of machine-observed features, and $k$ the bandwidth constraint.

\subsection{Hardness for Optimal Policies for Sophisticated Updating}

Even with binary features and the smallest nontrivial bandwidth $k=1$, optimizing a highlighting policy for a \emph{sophisticated} agent is NP-hard.
The following theorem formalizes our NP-hardness result by a reduction from Euclidean $2$-means clustering.

\begin{theorem}[NP-hardness of optimal highlighting for sophisticated agents]
\label{thm:nphard-sophisticated-opt}
Consider the decision-loss model with a sophisticated agent and bandwidth $k=1$.
Given an instance $(\mathbb P,A,\ell)$, computing an optimal highlighting policy
$$
\sigma^\star \in \argmin_{\sigma} R^S(\sigma)
$$
or equivalently, computing the optimal value $\min_\sigma R^S(\sigma)$ is NP-hard (in $n$ and $d$).
This remains true even when $\mathcal{X} \subseteq\{0,1\}^d$.
\end{theorem}

\begin{proof}[Proof sketch]
We reduce the problem from Euclidean sum-of-squares clustering with $K=2$, which is NP-hard, as stated by
\cite{aloise2009np}. Given points $z_1,\dots,z_m\in\mathbb R^p$, choose
$d=2+\lceil\log_2 m\rceil$ so $2^{d-2}\ge m$. Under bandwidth $k=1$, any deterministic policy induces at most
$|\mathcal S|=2d+1$ signals. Reserve two signals $s_1=(\{1\},0)$ and $s_2=(\{2\},0)$ for clustering.

Construct $m$ \emph{data} states $y^i=(0,0,\mathrm{enc}(i))\in\{0,1\}^d$ with distinct tails, and
$2d-1$ \emph{gadget} states $x^0,\dots,x^{2d-2}\in\{0,1\}^d$ with $x^t_1=x^t_2=1$, one for each
non-reserved signal, so each gadget can realize a distinct non-reserved signal under $k=1$.
Let $\mathbb P$ be uniform over the $n=m+(2d-1)$ states. Let $A=\mathbb R^p\cup\{0,1,\dots,2d-2\}$ and pick a large
constant $B$ larger than any feasible clustering cost bound. Define losses by
$\ell(t,x^t)=0$ and $\ell(a,x^t)=B$ for $a\neq t$, while for data states
$\ell(a,y^i)=\|a-z_i\|_2^2$ if $a\in\mathbb R^p$ and $\ell(a,y^i)=B$ otherwise.

With this construction, any policy that pools (i) two distinct gadgets or (ii) a gadget with a data state
induces conditional Bayes risk at least $B$ under that signal, so an optimal policy must assign all gadgets
to distinct non-reserved signals and must not mix gadgets with data states. Hence all data states must be
mapped to exactly the two reserved signals $s_1,s_2$, inducing a partition $(C_1,C_2)$ of $[m]$.
Under squared loss, the sophisticated best response to each signal is the posterior mean,
so the optimal highlighting value satisfies
$$
\min_{\sigma} R^S(\sigma)
\;=\;
\frac{1}{n}\min_{(C_1,C_2)}\ \sum_{t=1}^2\sum_{i\in C_t}\|z_i-\mu_t\|_2^2,
$$
where $\mu_t$ is the centroid of $\{z_i:i\in C_t\}$. Therefore, computing $\min_\sigma R^S(\sigma)$
computes the Euclidean $2$-means optimum, implying NP-hardness.
\end{proof}

\subsection{Tractability of Optimal Policies for Naive Updating}
\label{sec:algorithms}

The above result shows that determining the optimal highlighting policy for a sophisticated agent is NP-hard, even when $k=1$.
In this subsection, we characterize optimal highlighting for a naive agent.
The naive model exhibits a different computational complexity from the sophisticated model: in particular, the machine's optimization problem is polynomial-time solvable when $k$ is a fixed constant.

A key advantage of the naive model is computational.
Because the naive agent treats the selection event as exogenous, the machine can evaluate the realized loss under any candidate revealed set $I$ pointwise for each realized instance $x$ by computing the naive posterior $\P(\cdot\mid X_I=x_I)$ and the induced Bayes action.
When the bandwidth $k$ is bounded, this yields an exact algorithm via enumeration; we also present three scalable heuristics that will be useful in binary and correlated settings.

\paragraph{Exact optimization by enumeration.}
For each realized instance $x$, define the realized loss induced by revealing $(I,x_I)$ under the naive posterior:
$$
L^N(I;x)\;:=\;\ell\!\left(\hat a\!\left(\ell,\P(\cdot\mid X_I=x_I)\right),\,x\right).
$$
The naive--optimal policy at $x$ is obtained by minimizing $L^N(I;x)$ over all $|I|\le k$.
The exact optimization algorithm would enumerate all possible $I$ and calculate the loss (\autoref{alg:enum-naive}).
The following remark implies that if $k$ is small, then this algorithm can find the optimal policy in polynomial time.

\begin{algorithm}[h]
\caption{Exact enumeration}
\label{alg:enum-naive}
\KwInput{prior $\P$ on $X$ (discrete support), loss $\ell$, action set $A$, realized instance $x$, bandwidth $k$}
\KwOutput{highlight set $\sigma^{\mathrm{enum}}(x)\subseteq[d]$}
\DontPrintSemicolon

$I^\star \gets \emptyset$; \quad $L^\star \gets +\infty$;\;

\For{$I\subseteq[d]$ \textnormal{with} $|I|\le k$}{
    Compute naive posterior $\P_I(\cdot)\gets \P(\cdot \mid X_I=x_I)$;\;
    Compute Bayes action $\hat a_I \in \argmin_{a\in A}\ \E_{\P_I}[\ell(a,X)]$;\;
    $L \gets \ell(\hat a_I, x)$; \;
    \If{$L < L^\star$}{
        $L^\star \gets L$; \quad $I^\star \gets I$;\;
    }
}
\Return $\sigma^{\mathrm{enum}}(x)\gets I^\star$;\;
\end{algorithm}

\begin{remark}[Tractability of naive--optimal highlighting for bounded bandwidth]
\label{rm:naiveopt}
Consider the decision-loss model with a naive agent. Fix $k\in\mathbb{N}$.
Assume that, given $X_I$, one can evaluate the naive posterior $\P(\cdot \mid X_I)$ and compute a Bayes-optimal action
$
\hat a(\ell,\hat{\P}^{N}) \in \argmin_{a\in A}\; \E_{\hat{\P}^{N}}[\ell(a,X)]
$
in time $\mathrm{poly}(\mathrm{size}(\P,\ell,A))$.
Then an optimal highlighting policy for the naive agent can be computed in time
$
O\!\left(n \cdot \sum_{t=0}^{k} \binom{d}{t}\right)\cdot \mathrm{poly}(\mathrm{size}(\P,\ell,A)).
$
In particular, for constant $k$, computing a naive--optimal highlighting policy is polynomial-time in the input size.
\end{remark}

\subsection{Greedy-Type Algorithms}

When posterior updates are expensive, the model is high-dimensional, or the bandwidth is relatively large,
then complete enumeration may still be computationally unattractive.

\paragraph{A fast first-moment heuristic.}
An alternative for real-valued features $X_i$ is to highlight features whose realized values deviate most from their marginal means.
Let $\mu_i:=\E[X_i]$ and define $s_i(x):=|x_i-\mu_i|$.
The deviation heuristic selects the $k$ coordinates with the largest scores (\autoref{alg:dev}).
Indeed, the algorithm can be seen as using a surprise-based heuristic: it reports those components that represent the largest surprise relative to their expectation.
If two scores are equal, we break ties by choosing the one with the smallest index, which provides additional information for the sophisticated agent.
This rule ignores correlations and higher-order structure; nevertheless, in the independent Bernoulli recovery benchmark of \autoref{sec:binary}, it is optimal.

\begin{algorithm}[h]
\caption{Deviation heuristic}
\label{alg:dev}
\KwInput{marginal means $\mu_i=\E[X_i]$ for $i\in[d]$, realized instance $x$, bandwidth $k$}
\KwOutput{highlight set $\sigma^{\mathrm{dev}}(x)\subseteq[d]$}
\DontPrintSemicolon

\For{$i\gets 1$ \KwTo $d$}{
    $s_i \gets |x_i-\mu_i|$;\;
}
Let $I^\star$ be the indices of the $k$ largest values in $\{s_i\}_{i=1}^d$;\;
\Return $\sigma^{\mathrm{dev}}(x)\gets I^\star$;\;
\end{algorithm}

\paragraph{Marginal heuristic.}
Another fast contextual rule evaluates each feature only once, starting from the empty highlighted set.
For each coordinate $j\in[d]$, define the singleton marginal improvement
$$
\Delta(j\mid \emptyset,x)\;:=\;L^N(\emptyset;x)-L^N(\{j\};x).
$$
The marginal heuristic computes these one-step gains, ranks features by $\Delta(j\mid \emptyset,x)$, and highlights the top $k$ coordinates (\autoref{alg:marginal}).
Unlike the deviation heuristic, this rule can account for correlations through the posterior induced by revealing a single feature.
Unlike the greedy rule below, however, it does not update the ranking after earlier features have been selected, and therefore ignores overlap across the chosen set.

\begin{algorithm}[h]
\caption{Marginal heuristic}
\label{alg:marginal}
\KwInput{prior $\P$, loss $\ell$, action set $A$, realized instance $x$, bandwidth $k$}
\KwOutput{highlight set $\sigma^{\mathrm{marg}}(x)\subseteq[d]$}
\DontPrintSemicolon

$R_0 \gets L^N(\emptyset;x)$;\;
\For{$j\gets 1$ \KwTo $d$}{
    $\Delta_j \gets R_0 - L^N(\{j\};x)$;\;
}
Let $I^\star$ be the indices of the $k$ largest values in $\{\Delta_j\}_{j=1}^d$;\;
\Return $\sigma^{\mathrm{marg}}(x)\gets I^\star$;\;
\end{algorithm}

\paragraph{Greedy information gain.}

Another scalable design is a greedy policy that builds the highlighted set sequentially by maximizing the one-step reduction in the realized loss at the realized instance $x$, where the action is chosen as a Bayes action under the naive posterior.
For a candidate coordinate $j\notin I$, define the marginal improvement
$$
\Delta(j\mid I,x)\;:=\;L^N(I;x)-L^N(I\cup\{j\};x).
$$
Starting from $I=\emptyset$, the algorithm repeatedly adds the coordinate with the largest $\Delta(j\mid I,x)$, as long as this improvement is positive (\autoref{alg:greedy-ig}).

\begin{algorithm}[h]
\caption{Greedy information gain with early stopping}
\label{alg:greedy-ig}
\KwInput{prior $\P$, loss $\ell$, action set $A$, realized instance $x$, bandwidth $k$}
\KwOutput{highlight set $\sigma^{\mathrm{greedy}}(x)\subseteq[d]$}
\DontPrintSemicolon

$I \gets \emptyset$;\;
$\hat a_I \in \argmin_{a\in A}\ \E[\ell(a,X)\mid X_I=x_I]$;\;
$R \gets \ell(\hat a_I, x)$;\;

\For{$t\gets 1$ \KwTo $k$}{
    $j^\star \gets \bot$;\quad $\Delta^\star \gets -\infty$;\quad $R^\star \gets +\infty$;\;

    \For{$j\in[d]\setminus I$}{
        $\hat a_{I\cup\{j\}} \in \argmin_{a\in A}\ \E[\ell(a,X)\mid X_{I\cup\{j\}}=x_{I\cup\{j\}}]$;\;
        $R_j \gets \ell(\hat a_{I\cup\{j\}}, x)$;\;
        $\Delta_j \gets R - R_j$;\;

        \If{$\Delta_j > \Delta^\star$}{
            $\Delta^\star \gets \Delta_j$;\quad $j^\star \gets j$;\quad $R^\star \gets R_j$;\;
        }
    }
    \If{$\Delta^\star \le 0$}{
        \textbf{break};\tcp*{withhold further information}
    }
    $I \gets I\cup\{j^\star\}$;\quad $R \gets R^\star$;\;
}
\Return $\sigma^{\mathrm{greedy}}(x)\gets I$;\;
\end{algorithm}

\paragraph{Relationship among the algorithms.}
\autoref{alg:enum-naive} is exact and polynomial-time for constant $k$.
\autoref{alg:dev} is a fast heuristic that uses only first moments.
\autoref{alg:marginal} is an intermediate rule as it accounts for the posterior effect of revealing each feature individually, but it ignores redundancy across selected features because it never updates the ranking after a coordinate is added.
\autoref{alg:greedy-ig} is a greedy algorithm that accounts for belief updates and can exploit correlations: revealing one feature can change the posterior over others.
In the independent Bernoulli recovery benchmark (\autoref{sec:binary}), \autoref{alg:greedy-ig} and \autoref{alg:marginal} reduce to \autoref{alg:dev} and are optimal, so they lead to the same highlight sets as \autoref{alg:enum-naive}.

\section{Asymptotic Analysis for Binary Features}\label{sec:binary}

Above, we have considered general solutions for optimally highlighting features that apply for general features and loss functions.
There we have shown that optimal solutions as well as the computational complexity of finding them depend on how beliefs are formed.
In this section, we consider the specific case of binary features (i.e., $\mathcal{X}=\{0,1\}^d$) with the loss function given by the recovery loss $\ell(a,x) = \|a - x \|_2^2$.
Our goal is to characterize the asymptotic performance of the computationally attractive, simple greedy-type algorithms from the previous section.

In order to illustrate optimal solutions explicitly and characterize their performance, we first consider naive--optimal highlighting when features are independent.
We compare contextual highlighting to fixed subset selection, highlighting a fundamental difference between ex-ante and ex-post selection.
We argue that greedy-style algorithms can perform well and serve as candidate ``universal'' policies that remain competitive for sophisticated agents while being robust for naive agents.
Finally, we discuss extensions to weakly correlated settings.

\subsection{Closed-Form Optimal Policies for Naive Updating}

We first solve for optimal highlighting policies for naive updaters in the case when features are independent Bernoulli (but not necessarily identically distributed).
Specifically, let
$X_1,\dots,X_d$ be independent with $X_i\sim\mathrm{Bern}(p_i)$ and $p_i=\E[X_i]$.
Consider the squared recovery objective with $A=\R^d$ and
$
\ell(a,X)=\|a-X\|_2^2=\sum_{i=1}^d (a_i-X_i)^2.
$
If the machine reveals $(I,X_I)$, then under the naive posterior the Bayes action is
$$
\hat a_i=
\begin{cases}
X_i, & i\in I,\\
p_i, & i\notin I,
\end{cases}
$$
and therefore the realized loss at instance $X$ equals
\begin{equation}\label{eq:realized-loss-binary}
L^N(I;X)
=\sum_{i\notin I} (X_i-p_i)^2.
\end{equation}

\paragraph{Fixed subset benchmark (ex-ante selection).}
A fixed subset policy chooses a single set $I^\star$ with $|I^\star|\le k$ and reveals it for every instance.
Taking expectation in~\eqref{eq:realized-loss-binary} gives
$$
\E[L^N(I;X)]
=\sum_{i\notin I} \E[(X_i-p_i)^2]
=\sum_{i\notin I} \Var(X_i)
=\sum_{i\notin I} p_i(1-p_i).
$$
Hence, the optimal fixed subset reveals the $k$ coordinates with the largest variances $p_i(1-p_i)$.

\begin{lemma}[Optimal fixed subset under independence]\label{lm:fixed-indep}
In the independent Bernoulli recovery model, an optimal fixed subset policy chooses
$$
I^{\mathrm{fix}}\in\argmax_{I:\,|I|\le k}\ \sum_{i\in I} p_i(1-p_i),
$$
i.e., the $k$ coordinates with largest $p_i(1-p_i)$.
\end{lemma}

\paragraph{Contextual highlighting (ex-post selection).}
A contextual policy may choose $I$ as a function of the realized instance $X$.
Because the naive realized loss~\eqref{eq:realized-loss-binary} is additive across coordinates and depends on each coordinate only through $(X_i-p_i)^2$, the per-instance optimal rule is immediate.

\begin{lemma}[naive--optimal contextual highlighting under independence]\label{lm:contextual-indep}
In the independent Bernoulli recovery model, the naive--optimal contextual policy is
$$
\sigma^N(x)\in\argmax_{I:\,|I|\le k}\ \sum_{i\in I} (x_i-p_i)^2,
$$
i.e., it reveals the $k$ coordinates with the largest ``surprises'' $(x_i-p_i)^2$.
\end{lemma}

\paragraph{Contextual can beat fixed subset.}
Since a fixed subset policy is a special case of a contextual policy, it follows immediately that contextual highlighting weakly dominates fixed subset selection:
$$
R^N(\sigma^N) =
\E\!\left[\min_{I:\,|I|\le k}\ \sum_{i\notin I} (X_i-p_i)^2\right]
\ \le\
\min_{I:\,|I|\le k}\ \E\!\left[\sum_{i\notin I} (X_i-p_i)^2\right]
=R^N(\sigma^{\mathrm{fix}})
.
$$
Moreover, the inequality is strict whenever the identity of the $k$ largest surprise terms is not almost surely constant.
Intuitively, fixed subset selection hedges against uncertainty in expectation, whereas contextual selection exploits instance-level variation by revealing whichever coordinates happen to be most informative for that particular $x$.

\paragraph{Algorithmic implication.}
In the independence setting, by \autoref{lm:contextual-indep}, \autoref{alg:dev} (deviation) is optimal.
Furthermore, \autoref{alg:marginal} (marginal heuristic) and \autoref{alg:greedy-ig} (greedy information gain) coincide with \autoref{alg:dev}: revealing coordinate $i$ eliminates the term $(x_i-p_i)^2$ from~\eqref{eq:realized-loss-binary}, so the best marginal improvement is attained by the largest current surprise.

\subsection{Asymptotic Characterization of Average Losses}

We now solve for the asymptotic risk of the above policy for a naive and for a sophisticated updater.
In order to calculate approximate losses under the above policies, we consider sequences of distributions \(\P_d\)
over \(X_d = (X_{d,1},\ldots,X_{d,d})\)
with independent coordinates for each \(d\).
Each such distribution \(\P_d\) is given by the marginal probabilities
\begin{align*}
    p_d &= (p_{d,1}, \ldots, p_{d,d})
    \in [0,1]^d,
    &
    \E_d[X_{d,i}] &= p_{d,i}.
\end{align*}
We are specifically interested in the case where highlighting matters in a limit. Therefore,
we assume that the bandwidth \(k\) itself increases with the dimension:

\begin{assumption}[Proportional bandwidth]
\label{asm:bandwidth}
The bandwidth $k$ satisfies
\(k/d \rightarrow \alpha \in (0,1)\) as \(d \rightarrow \infty\).
\end{assumption}

In order to derive closed-form expressions for the expected loss under the above policies,
we assume that the sequence $(p_d)_{d \geq 1}$ is stable in the sense that the distribution over the $p_{d,i}$ themselves converges to a distribution over the unit interval with CDF $F$:

\begin{assumption}[Stability of the sequence of distributions]
\label{asm:stability}
    There is a monotone increasing c\`adl\`ag function $F: [0,1] \rightarrow [0,1]$
    with $\lim_{t \rightarrow 0} F(t) = 0, \lim_{t \rightarrow 1} F(t) = 1$
    such that
    \begin{equation}
        \frac{1}{d} \sum_{i=1}^d \mathbbm{1}(p_{d,i} \leq t) \rightarrow F(t)
        \label{eqn:stability}
    \end{equation}
    as $d \rightarrow \infty$
    for all continuity points $t$ of $F$.
\end{assumption}

This assumption leaves the $p_{d,i}$ flexible, and merely assumes that their distribution across $i$ can be approximated well by some function $F$.

For formulating our results, it will be convenient to transform the distributions $p_d$ and realizations $X_d$.
Specifically, for all $d, i, t$, let
\begin{align*}
    p^*_{d,i} &= \begin{cases}
        p_{d,i}, &p_{d,i} \leq \frac{1}{2}, \\
        1 {-} p_{d,i}, & \text{otherwise}
    \end{cases}
    =
    \min(p_{d,i}, 1 {-} p_{d,i}),
    \\
    X^*_{d,i} &= \begin{cases}
        X_{d,i}, &p_{d,i} \leq \frac{1}{2}, \\
        1 {-} X_{d,i}, & \text{otherwise}
    \end{cases}
    = \begin{cases}
        X_{d,i}, &p_{d,i} = p^*_{d,i}, \\
        1 {-} X_{d,i}, & \text{otherwise},
    \end{cases}
    \\
    F^*(t) &= \begin{cases}
        F(t) + 1 {-} F(1{-}t), &t \leq \frac{1}{2},
        \\
        1, & \text{otherwise}.
    \end{cases}
\end{align*}
Then $(p_{d,i}^*)_{d,i}, F^*$ fulfills \eqref{eqn:stability} with $F^*$ continuously and (weakly) monotonically increasing.
Furthermore, by the symmetry of the loss function, we can focus on highlighting the transformed features $X^*_d$.

We now formulate our main policies of interest in terms of the transformed quantities $p^*_d, X^*_d$.
For every $d$ consider a (fixed) ordering
$(p^*_{d,(1)},\ldots, p^*_{d,(d)})$ of $p^*$ such that the $p^*_{d,(j)}$ are monotonically increasing, with corresponding realizations $X^*_{d,(j)} \sim \text{Bern}(p^*_{d,(j)})$.
Then the two highlighting choices of interest can, without loss of generality, be written as:
\begin{enumerate}[label=(\alph*)]
    \item \emph{Fixed ex-ante highlighting}:
    An optimal \emph{fixed} set for highlighting of size $k$ for the naive agent consists of $X^*_{d,(d - k + 1)}, \ldots, X^*_{d,(d)}$, corresponding to the $k$ highest probabilities $p^*_{d,i}$, which correspond exactly to those original cases where $p_{d,i} (1- p_{d,i})$ is maximal.
    \label{enum:fixed}
    \item \emph{Greedy data-driven highlighting} (Algorithms~\ref{alg:enum-naive}, \ref{alg:dev}, \ref{alg:marginal}, \ref{alg:greedy-ig}):
    An optimal \emph{data-driven} set for highlighting of size at most $k$ consists of the
    $X^*_{d,(j)}$
    for the first $k$ indices $j$ for which $X^*_{d,(j)} = 1$.%
    \footnote{If there are less than $k$ such indices, then we add the highest indices $j$ for which $X^*_{d,(j)} = 0$.}
    These correspond to the features with the $k$ highest original values of $(X_{d,i} - p_{d,i})^2$.
    \label{enum:greedy}
\end{enumerate}
For developing our results, it will be helpful to consider a third highlighting scheme:
\begin{enumerate}[label=(\alph*),resume]
    \item \emph{Fixed surprise highlighting}: For a fraction \(\beta \in (0,1)\), highlight those among the features \(X^*_{d,(1)}, \ldots, X^*_{d,(\lfloor \beta d \rfloor)}\) with \(X^*_{d,(j)} = 1\), corresponding to a varying number of highlighted features (that may be lower or higher than \(k\), and can thus be infeasible).
    \label{enum:fraction}
\end{enumerate}
To establish the following results, we show that \autoref{enum:greedy} and \autoref{enum:fraction} are asymptotically equivalent for the right value of $\beta$, which also allows us to characterize their asymptotic gap to \autoref{enum:fixed}.
In order to state these results,
it will be helpful to consider the quantile function
\begin{align*}
    Q^*(q) = \inf\{t \in [0,1]; F^*(t) \geq q\}
\end{align*}
corresponding to the cdf $F^*$, which is its generalized inverse.
Based on $Q^*$, we first characterize the asymptotic risk for fixed highlighting, where the highlighted set is chosen ex-ante.

\begin{proposition}[Asymptotic risk of optimal fixed highlighting]
\label{prop:fixedasymptotic}
Assume that Assumptions~\ref{asm:bandwidth} and \ref{asm:stability} hold.
Then, under optimal fixed highlighting (\autoref{enum:fixed}), both the naive and the sophisticated agent have the same asymptotic normalized risk, given by
\[
\lim_{d\to\infty} R^N(\sigma^{\ref{enum:fixed}})
=
\lim_{d\to\infty} R^S(\sigma^{\ref{enum:fixed}})
=
\lim_{d\to\infty}
\E\left[
\frac{1}{d} \sum_{j=1}^{d-k}
\bigl(X^*_{d,(j)}-p^*_{d,(j)}\bigr)^2
\right]
=
\int_{0}^{1-\alpha} Q^*(q)\bigl(1-Q^*(q)\bigr)\,\de q.
\]
\end{proposition}

We next provide an explicit expression for highlighting surprises based on a fixed fraction of features.

\begin{proposition}[Asymptotic risk of fixed surprise highlighting]
\label{prop:fractionasymptotic}
Assume that Assumptions~\ref{asm:bandwidth} and \ref{asm:stability} hold, and consider \autoref{enum:fraction} for a given \(\beta \in (0,1)\).

\begin{enumerate}[label=(\roman*)]
\item \emph{Sophisticated agent}: The asymptotic normalized risk is
\[
\lim_{d\to\infty} R^S(\sigma^{\ref{enum:fraction}}_\beta)
=
\lim_{d\to\infty}
\E\left[
\frac{1}{d} \sum_{j=\lfloor \beta d\rfloor + 1}^{d}
\bigl(X^*_{d,(j)}-p^*_{d,(j)}\bigr)^2
\right]
=
\int_{\beta}^{1} Q^*(q)\bigl(1-Q^*(q)\bigr)\,\de q.
\]
\item \emph{Naive agent}: The asymptotic normalized risk is
\begin{align*}
\lim_{d\to\infty} R^N(\sigma^{\ref{enum:fraction}}_\beta)
&=
\lim_{d\to\infty}
\E\left[
\frac{1}{d}\sum_{j=1}^{\lfloor \beta d\rfloor}
\Ind\bigl(X^*_{d,(j)}=0\bigr)\bigl(X^*_{d,(j)}-p^*_{d,(j)}\bigr)^2
+
\frac{1}{d} \sum_{j=\lfloor \beta d\rfloor + 1}^{d}
\bigl(X^*_{d,(j)}-p^*_{d,(j)}\bigr)^2
\right]
\\
&=
\int_{0}^{\beta} \bigl(Q^*(q)\bigr)^2\bigl(1-Q^*(q)\bigr)\,\de q
+
\int_{\beta}^{1} Q^*(q)\bigl(1-Q^*(q)\bigr)\,\de q.
\end{align*}
\end{enumerate}
\end{proposition}

Assuming that the asymptotic bandwidth is not too large, we can find a fraction $\beta$ such that Procedures~\ref{enum:greedy} and \ref{enum:fraction} coincide.

\begin{assumption}[Limited asymptotic bandwidth]
\label{asm:smallbandwidth}
    The limiting bandwidth fraction satisfies
    \[
        \alpha \in \left(0,\ \int_{0}^{1} Q^*(q)\,\de q\right).
    \]
\end{assumption}

\begin{proposition}[Asymptotic equivalence of greedy and fixed surprise highlighting]
\label{prop:equivalence}
Assume that Assumptions~\ref{asm:bandwidth}, \ref{asm:stability}, \ref{asm:smallbandwidth} hold.
Let \(\beta^*\) satisfy
\(
\int_{0}^{\beta^*} Q^*(q)\,\de q = \alpha.
\)
Then the greedy procedure (\autoref{enum:greedy}) and the fixed surprise-highlighting procedure (\autoref{enum:fraction}) with parameter \(\beta^*\) have the same asymptotic normalized risk for both sophisticated and naive agents.
In particular, their asymptotic normalized risks are the expressions in \autoref{prop:fractionasymptotic} evaluated at \(\beta=\beta^*\).
\end{proposition}

Putting these results together, we obtain an exact characterization of the limiting risks along our asymptotic sequence.

\begin{theorem}[Asymptotic risk of the greedy algorithm]
\label{thm:greedyasymptotic}
Assume that Assumptions~\ref{asm:bandwidth}, \ref{asm:stability}, \ref{asm:smallbandwidth} hold.
Let \(\beta^*\) be as defined in \autoref{prop:equivalence}.

\begin{enumerate}[label=(\roman*)]
\item \emph{Sophisticated agent}: The greedy \autoref{enum:greedy} has asymptotic normalized risk
\[
\lim_{d\to\infty} R^S(\sigma^{\ref{enum:greedy}})
=
\int_{\beta^*}^{1} Q^*(q)\bigl(1-Q^*(q)\bigr)\,\de q.
\]
\item \emph{Naive agent}: The greedy \autoref{enum:greedy} has asymptotic normalized risk
\[
\lim_{d\to\infty} R^N(\sigma^{\ref{enum:greedy}})
=
\int_{0}^{\beta^*} \bigl(Q^*(q)\bigr)^2\bigl(1-Q^*(q)\bigr)\,\de q
+
\int_{\beta^*}^{1} Q^*(q)\bigl(1-Q^*(q)\bigr)\,\de q.
\]
\end{enumerate}
\end{theorem}

\autoref{fig:exante-vs-contextual} illustrates the qualitative difference between \emph{ex-ante} fixed subset selection and \emph{contextual} (data-driven) highlighting.
We consider a sequence of independent Bernoulli features $X_{d,i}\sim\mathrm{Bern}(p_{d,i})$ whose empirical distribution of means converges to the symmetric triangular density
$f(p)=4p$ for $p\in[0,\tfrac12]$ and $f(p)=4-4p$ for $p\in[\tfrac12,1]$.
Fix $\alpha=\tfrac14$ so that $k= d / 4$.

Under an optimal \emph{ex-ante} fixed policy (\autoref{enum:fixed}), the designer chooses one subset of size $k$ before observing $X_d$; for squared-error recovery, this amounts to prioritizing features with the largest intrinsic uncertainty $p_i(1-p_i)$, i.e., those with $p_i$ closest to $\tfrac12$ (blue region).
In contrast, the optimal \emph{contextual} policy (\autoref{enum:greedy}) conditions on the realized instance and highlights the coordinates with the largest realized deviation $|X_i-p_i|$.
As $d\to\infty$, this ``surprise-driven'' rule concentrates highlighting on extreme means, producing the triangular regions shown in orange.

\begin{figure}[t]
    \centering
    \includegraphics[width=\textwidth]{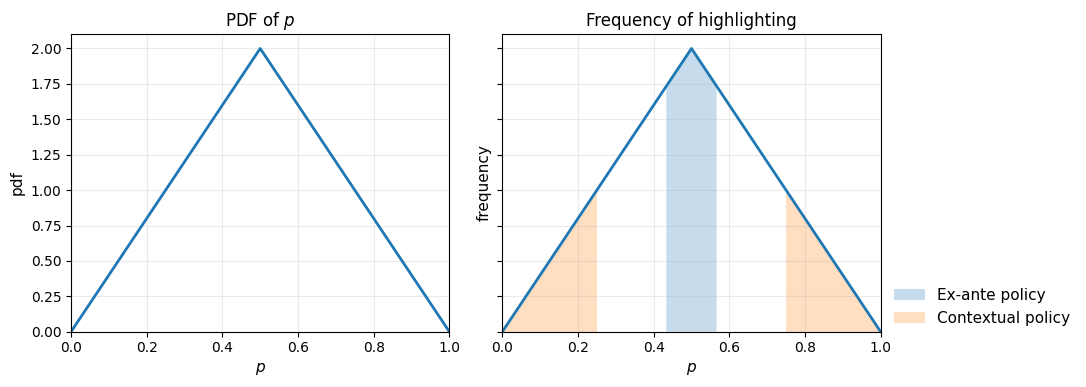}
    \caption{Ex-ante fixed subset selection vs.\ contextual highlighting.
    \emph{Left:} the limiting density of feature means $p$, given by
    $f(p)=4p$ for $p\in[0,\tfrac12]$ and $f(p)=4-4p$ for $p\in[\tfrac12,1]$.
    \emph{Right:} the asymptotic highlighting frequency of a feature with mean $p$ under an optimal contextual (data-driven) highlighting policy with bandwidth fraction $\alpha=0.25$ (i.e., $k=\alpha d$).
    The blue region indicates the ex-ante fixed policy, which prioritizes features with largest variance $p(1-p)$ (hence $p$ near $\tfrac12$).
    The orange triangles indicate the contextual policy, which highlights features with the largest realized surprise $|X_i-p_i|$, leading asymptotically to emphasizing the extreme-$p$ regions.}
    \label{fig:exante-vs-contextual}
\end{figure}

A simpler numerical example is given by the case where all probabilities are the same, $p_{d,i}=p$ for all $d$ and $i$, so that the features are i.i.d.
We assume without loss of generality that $p \leq \frac{1}{2}$.
Then $Q^*(q)=p$ for all $q\in(0,1)$.
By \autoref{prop:fixedasymptotic}, the asymptotic normalized risk of the fixed highlighting scheme equals
\[
\lim_{d\to\infty} R^N(\sigma^{\ref{enum:fixed}})
=
\lim_{d\to\infty} R^S(\sigma^{\ref{enum:fixed}})
=
\int_{0}^{1-\alpha} Q^*(q)\bigl(1-Q^*(q)\bigr)\,\de q = (1-\alpha)\,p(1-p).
\]
By \autoref{prop:equivalence}, $\beta^*$ solves
$\int_{0}^{\beta^*} Q^*(q)\,\de q = \beta^* p = \alpha$, so $\beta^*=\alpha/p$ (and we assume that $\alpha < p$ by \autoref{asm:smallbandwidth}).
The greedy risk has normalized asymptotic risk
\[
\lim_{d\to\infty} R^S(\sigma^{\ref{enum:greedy}})
=
\int_{\beta^*}^{1} Q^*(q)\bigl(1-Q^*(q)\bigr)\,\de q = (1-\beta^*)\,p(1-p) = (p-\alpha)(1-p)
\]
for the sophisticated agent, and asymptotic normalized risk
\begin{align*}
    \lim_{d\to\infty} R^N(\sigma^{\ref{enum:greedy}})
    &=
    \int_{0}^{\beta^*} \bigl(Q^*(q)\bigr)^2\bigl(1-Q^*(q)\bigr)\,\de q + \int_{\beta^*}^{1} Q^*(q)\bigl(1-Q^*(q)\bigr)\,\de q
    \\
    &= \beta^*\,p^2(1-p) + (1-\beta^*)\,p(1-p) = \alpha p(1-p) + (p - \alpha)(1-p)
\end{align*}
for the naive agent.
As $\alpha \rightarrow p$, the sophisticated agent achieves an average risk approaching zero, the naive agent achieves an average risk $p^2 (1-p)$, and a fixed highlighting scheme achieves an average risk $p (1-p)^2$.

\subsection{Partially Correlated Features}\label{sec:correlated}

In the previous subsection, we focused on independent features, where it is optimal for a naive receiver to order features by their marginal information value and reveal as many features as the bandwidth constraint allows.
With correlation, however, the value of highlighting also depends on its implication for unrevealed features.
For a naive receiver who ignores the selection rule, revealing an additional feature can even increase the realized loss in some instances.
The next example shows that, in correlated settings, a naive optimal contextual policy need not be monotone in the amount of information highlighted.
In particular, it may be optimal to withhold information in some cases.

\begin{example}[With correlation, revealing more can hurt a naive agent]\label{ex:correlated-nonmonotone}
Let $X=(X_1,X_2)\in\{0,1\}^2$ with correlated prior
$
\P(0,0)=\P(0,1)=\P(1,0)=\frac13, \P(1,1)=0.
$
The machine observes $X$ and may reveal at most one coordinate (bandwidth $k=1$), i.e., it chooses $I\in\{\emptyset,\{1\},\{2\}\}$ and $(I,X_I)$ is revealed.
The receiver is naive: upon observing $(I,X_I)$, they use the posterior $\P(\cdot\mid X_I)$ (ignoring the selection rule).
Consider squared recovery of the full vector,
$
A=\R^2, \ell(a,x)=\|a-x\|_2^2.
$
At the instance $X=(0,0)$, revealing $X_1=0$ induces the naive posterior mean $\E[X_2\mid X_1=0]=\frac{1}{2}$, hence the naive prediction is $(0,\frac{1}{2})$ and the realized loss is $\frac{1}{4}$.
Symmetrically, revealing $X_2=0$ yields prediction $(\frac{1}{2},0)$ and loss $\frac{1}{4}$.
In contrast, revealing nothing induces the prior-mean prediction $\E[X]=(\frac{1}{3},\frac{1}{3})$ and realized loss $\frac{2}{9}<\frac{1}{4}$.
Thus, when features are correlated, highlighting additional information need not improve performance for a naive receiver, and an optimal contextual policy may strictly prefer to withhold information in some cases.
\end{example}

\autoref{ex:correlated-nonmonotone} also illustrates why the deviation heuristic (\autoref{alg:dev}) and the marginal heuristic (\autoref{alg:marginal}) can be suboptimal under correlation, whereas greedy information gain (\autoref{alg:greedy-ig}) can adapt by withholding information when the best available reveal would increase the realized loss.
In \autoref{ex:correlated-nonmonotone} with $k=1$, the marginal means are $\mu_1=\mu_2=\E[X_i]=1/3$.
\autoref{alg:dev} reveals one coordinate at $X=(0,0)$, and incurs loss $1/4$ there, while achieving zero loss at $X=(0,1)$ and $X=(1,0)$.
Hence
$$
R^N(\sigma^{\mathrm{dev}})=\frac13\cdot\frac14=\frac{1}{12}.
$$
By contrast, greedy information gain with early stopping selects $I=\emptyset$ at $X=(0,0)$ (loss $2/9<1/4$) and reveals the ``$1$'' coordinate in the other two states, achieving zero loss there.
Therefore
$$
R^N(\sigma^{\mathrm{greedy}})=\frac13\cdot\frac{2}{9}=\frac{2}{27}
\;<\;
\frac{1}{12}.
$$

The above examples show cases where 
$\sigma^{\mathrm{greedy}}$ is a better selection rule than other approximate algorithms. However, the following remark indicates that sometimes \autoref{alg:greedy-ig} can perform worse than \autoref{alg:dev} and \autoref{alg:marginal}.

\begin{remark}[\autoref{alg:greedy-ig} can do worse than \autoref{alg:dev} and \autoref{alg:marginal}]\label{rem:greedy-worse}
The early-stopping rule that allows \autoref{alg:greedy-ig} to withhold individually harmful reveals can also cause it to reveal strictly too little when reveals help only jointly. Consider \autoref{ex:correlated-nonmonotone} with bandwidth $k=2$. At $X=(0,0)$, every singleton reveal raises the realized loss: $\Delta(j\mid\emptyset,X)=\tfrac{2}{9}-\tfrac{1}{4}=-\tfrac{1}{36}<0$ for $j\in\{1,2\}$, so \autoref{alg:greedy-ig} stops at $\emptyset$ with loss $2/9$.
Both \autoref{alg:dev} and \autoref{alg:marginal} reveal $\{1,2\}$ at $(0,0)$, which pins down the state and attains loss $0$. At the other two realizations $X=(0,1)$ and $X=(1,0)$, all three algorithms achieve loss $0$. Therefore, in expectation across the three equally likely realizations,
$$
R^N(\sigma^{\mathrm{greedy}}) = \tfrac{1}{3}\cdot\tfrac{2}{9} = \tfrac{2}{27} \;>\; 0 \;=\; R^N(\sigma^{\mathrm{dev}}) = R^N(\sigma^{\mathrm{marg}}),
$$
so the deviation and marginal heuristics strictly dominate greedy information gain in expected naive loss.
\end{remark}

The previous remark relies on \autoref{alg:greedy-ig}'s early-stopping rule. Even without early stopping, \autoref{alg:greedy-ig}'s myopic first-step choice can lock it out of the optimal pair, as the following discrete example shows.

\begin{remark}[\autoref{alg:greedy-ig} without early stopping can be worse than \autoref{alg:dev}]\label{rem:greedy-worse-no-stop}
Let $X=(X_1,X_2,X_3)\in\{0,1\}^3$ with bandwidth $k=2$ and four equally likely realizations
$$
r_1=(0,0,0),\quad r_2=(1,0,0),\quad r_3=(1,1,0),\quad r_4=(0,1,1),
$$
so $\mu=(\tfrac{1}{2},\tfrac{1}{2},\tfrac{1}{4})$. 
At $X=r_2$, 
revealing $X_3=0$ shifts the naive posterior to $\E[X_1\mid X_3=0]=2/3$ and $\E[X_2\mid X_3=0]=1/3$, close to $x_1=1$ and $x_2=0$. \autoref{alg:greedy-ig} therefore picks $X_3$ first; conditional on $\{3\}$, both remaining additions have $\Delta(j\mid\{3\},r_2)=-\tfrac{1}{36}$, so even without early stopping the algorithm completes at $\sigma^{\mathrm{greedy}}(r_2)=\{1,3\}$ with realized loss $\tfrac{1}{4}$. \autoref{alg:marginal} also selects $\{1,3\}$ from the singleton ranking. \autoref{alg:dev}, by contrast, ranks coordinates by $|x-\mu|=(\tfrac{1}{2},\tfrac{1}{2},\tfrac{1}{4})$ and picks $\{1,2\}$, which uniquely identifies $r_2$ and attains loss $0$. At the other three realizations $r_1, r_3, r_4$, all three algorithms attain loss $0$. Therefore, in expectation across the four equally likely realizations,
$$
R^N(\sigma^{\mathrm{greedy}})=R^N(\sigma^{\mathrm{marg}})=\tfrac{1}{4}\cdot\tfrac{1}{4}=\tfrac{1}{16}>0=R^N(\sigma^{\mathrm{dev}}),
$$
so the deviation heuristic strictly dominates both greedy information gain and the marginal heuristic in expected naive loss. 
\end{remark}

\paragraph{Good approximation for the correlation case under a weak-mean-shift assumption.}

A simple sufficient assumption is that conditioning on a small number of revealed coordinates does not substantially change posterior means of unrevealed coordinates.

\begin{assumption}[Weak-mean-shift]\label{asm:weak-corr-mean}
 Assume there exists $\varepsilon\ge 0$ such that for every $i\in[d]$, every set $S\subseteq[d]\setminus\{i\}$ with $|S|\le k$, and every value $x_S$,
\begin{equation}
\left|\E[X_i\mid X_S=x_S]-\E[X_i]\right|\le \varepsilon.
\end{equation}
\end{assumption}
\autoref{asm:weak-corr-mean} says that even after observing up to $k$ coordinates, each remaining coordinate's conditional mean stays within $\varepsilon$ of its unconditional mean.
This holds, for example, under standard weak-dependence or bounded-influence regimes.

\begin{example}[Two sufficient conditions for weak-mean-shift]\label{lem:wc-suff}
Assume $X_i\in[a_i,b_i]$ almost surely and let $R:=\max_i(b_i-a_i)$.
Fix $i$ and $S$ with $|S|\le k$.
\begin{enumerate}[label=(\roman*)]
\item Distributional stability in total variation.
If
$$
\sup_{x_S}\ \mathrm{TV}\!\big(\mathcal{L}(X_i\mid X_S=x_S),\ \mathcal{L}(X_i)\big)\ \le\ \delta,
$$
then $\sup_{x_S}\big|\E[X_i\mid X_S=x_S]-\E[X_i]\big|\le \delta R$.

\item Linear conditional mean with bounded sparse regression.
Suppose there exist coefficients $\beta_{iS}\in\R^{|S|}$ such that
$$
\E[X_i\mid X_S=x_S]=\mu_i+\beta_{iS}^\top(x_S-\mu_S)
\quad\text{for all }x_S,
$$
and $\|\beta_{iS}\|_1\le \eta$.
If additionally $\|x_S-\mu_S\|_\infty\le R$ almost surely, then
$\sup_{x_S}\big|\E[X_i\mid X_S=x_S]-\mu_i\big|\le \eta R$.
\end{enumerate}
\end{example}

 The simple deviation heuristic is close to optimal under the above assumption, provided that the bound $\varepsilon$ is sufficiently small.

\begin{proposition}[Deviation heuristic is near-optimal under weak correlation]\label{prop:dev-weak-corr-general}
Assume $X\in\R^d$ and the objective is squared recovery.
Let $\mu_i=\E[X_i]$.
Assume \autoref{asm:weak-corr-mean} is satisfied.
Assume additionally that $X_i\in[a_i,b_i]$ almost surely, and let $R:=\max_i (b_i-a_i)$.
Let $\sigma^{\mathrm{dev}}$ be the policy from \autoref{alg:dev}, and let $\sigma^N$ be a naive--optimal policy among all policies that reveal at most $k$ coordinates (and may reveal none).
Then
$$
R^N(\sigma^{\mathrm{dev}})
\ \le\
R^N(\sigma^N)\;+\;2(d-k)\cdot(2R\varepsilon+\varepsilon^2).
$$
In particular, when $\varepsilon$ is small, the deviation heuristic is near-optimal up to an additive error that vanishes as $\varepsilon\to 0$.
\end{proposition}

\section{Simulation Case Study Based on the American Housing Survey}\label{sec:empirical}

We now illustrate our theoretical results in a calibrated numerical exercise using data from the American Housing Survey (AHS).
Home valuation is a canonical decision with many potentially relevant attributes: a prospective buyer or appraiser cares about structural characteristics, neighborhood quality, and price, yet typically lacks the bandwidth to inspect every detail of every listing.
In such settings, an algorithm can support the decision-maker by highlighting a small subset of attributes that are most informative for getting a realistic impression and assessing the value of each home.

For our specific application, we leverage data from 1,521 homes in the Minneapolis--St.~Paul--Bloomington metro area, model the human's prior as a multivariate Gaussian over 44 housing characteristics, and compare a range of highlighting policies primarily from the perspective of a naive agent, with a supplementary robustness exercise for a stylized sophisticated agent.
The exercise yields three main takeaways that mirror the theory:
(1) contextual highlighting delivers substantial gains over fixed, ex-ante highlighting;
(2) taking covariance into account is important, since there is otherwise substantive redundancy across selected features;
and (3) it can sometimes be optimal to highlight less rather than more, as revealing additional features can increase the loss faced by a naive agent in correlated environments.

\subsection{Setup and Data}

We use the pooled 2021 and 2023 AHS metropolitan samples, restricting to owner-occupied units with an observed self-reported market value (\texttt{MARKETVAL}) in the Minneapolis--St.~Paul--Bloomington metro area ($n=1{,}521$). The state vector has 44 features, including $\log(\texttt{MARKETVAL})$, structural characteristics (unit size, stories, lot size, number of bedrooms, bathrooms, and total rooms), building type, ownership status, heating and cooling systems, condition indicators, and neighborhood quality ratings.
For simplicity, we treat all features as real values.

\subsection{Belief Model and Loss}

We model the state as a multivariate Gaussian over the 44 features, with the prior mean and covariance estimated from the data. The price coordinate is stored as $\log(\texttt{MARKETVAL})$; its prior mean is 12.6725, corresponding to approximately \$319{,}000 in the original dollar scale. When a set $S$ of features is revealed, the naive agent updates to the Gaussian posterior conditional on $X_S$, leaving unrevealed coordinates at their conditional expectations.

We use a loss function that is equivalent to the outcome-targeted recovery loss from \autoref{sec:model}.
Specifically, for $Y = \log(\texttt{MARKETVAL})$, we use the loss function
\begin{equation}\label{eq:ahs-loss-instance}
\ell(\hat{x}, \hat{y}; x, y) = \alpha \, \frac{(\hat{y}-y)^2}{\sigma_y^2} + (1-\alpha) \, \frac{\sum_{j} w_j \, (\hat{x}_j - x_j)^2/\sigma_j^2}{\sum_{j} w_j},
\end{equation}
where $w_j$ are feature-specific importance weights and $\sigma_j^2$ is the sample variance of each variable. The weight on $\log(\texttt{MARKETVAL})$ is fixed exogenously at $\alpha=\tfrac{1}{2}$; the remaining weights are derived from the coefficients of a ridge regression of $\log(\texttt{MARKETVAL})$ on the other features ($R^2=0.51$).
We then evaluate choices by the average sample loss
\begin{equation}\label{eq:ahs-loss}
\mathcal{L}
\;=\;
\frac{1}{n} \sum_{i=1}^{n} \ell(\hat{x}_i, \hat{y}_i; x_i, y_i)
\;=\;
\frac{1}{n} \sum_{i=1}^{n}
\left[
\alpha \, \frac{(\hat{y}_i - y_i)^2}{\sigma_y^2}
+ (1-\alpha) \, \frac{\sum_{j \in \text{obs}(i)} w_j \, (\hat{x}_{ij} - x_{ij})^2/\sigma_j^2}{\sum_{j \in \text{obs}(i)} w_j}
\right].
\end{equation}
By construction, the no-reveal benchmark achieves a loss of exactly one. All results below refer to the main specification in which \texttt{MARKETVAL} is not directly revealable.

\subsection{Highlighting Policies}

We compare several policies, organized along two margins emphasized in the theory: fixed versus contextual highlighting, and covariance-aware versus covariance-ignorant selection.

\begin{itemize}
    \item \emph{Fixed top-$k$ importance.} The simplest ex-ante benchmark. It ignores covariance, ranks features once by loss weight, and reveals the same top-$k$ non-price features for every house.

    \item \emph{Fixed marginal value.} An intermediate ex-ante benchmark that accounts for covariance but not for redundancy among the selected features. For each revealable feature~$j$, it computes the aggregate training loss when \emph{only}~$j$ is revealed, using the full covariance structure to update the posterior over all other features. It then ranks features by marginal loss reduction and selects the top~$k$.

    \item \emph{Fixed forward stepwise.} The main ex-ante covariance-aware benchmark. Before seeing any specific house, it builds a single ordered list of features by greedy forward selection. At each step, it adds the feature that produces the largest average reduction in the training objective given the features already selected, so it naturally avoids selecting redundant features.

    \item \emph{Fixed smart greedy.} The ex-ante analogue of smart greedy. It follows the same forward-stepwise search, but may stop early if the next feature does not improve the aggregate fixed-policy objective. In the current sample, this stopping rule is not binding at the budgets we report, so its performance is numerically identical to forward stepwise.

    \item \emph{Contextual deviation heuristic} (\autoref{alg:dev}). The empirical counterpart of the surprise-based rule. For each realized house, it reveals the $k$ observed non-price features with the largest absolute deviation from their marginal means. The rule uses only first moments and ignores covariance, so it is the fastest contextual benchmark but also the most restrictive informationally.

    \item \emph{Contextual marginal top-$k$} (\autoref{alg:marginal}). This policy reacts to the realized house, but scores features one at a time and ignores overlap among selected features. It is therefore contextual but covariance-ignorant.

    \item \emph{Contextual greedy.} A variant of \autoref{alg:greedy-ig} without early stopping. For each house, it chooses features sequentially. At each step, it selects the feature with the largest one-step reduction in that row's loss and then updates the posterior.

    \item \emph{Smart greedy} (\autoref{alg:greedy-ig}). Identical to contextual greedy except that it may stop early if the next step does not improve the current row loss. At small budgets, it is nearly indistinguishable from contextual greedy. At large budgets, early stopping can improve performance by avoiding features whose Gaussian conditional updates introduce more noise than information (see below).
\end{itemize}

For small bandwidths we also compute exact enumeration benchmarks. On the fixed side, we enumerate all subsets for $k \leq 3$ to obtain the optimal ex-ante subset under the current objective. On the contextual side, we also run exact enumeration (\autoref{alg:enum-naive}) for $k \leq 3$ on the full one-metro sample. This is computationally feasible in the current application and puts the exact contextual benchmark on the same footing as the other policies.

\subsection{Results}

\autoref{tab:ahs-hidden-main} reports the main results, combining all fixed and contextual policies in a single table to facilitate comparison across the two margins.

\begin{table}[h]
\centering
\caption{Global-normalized recovery loss with \texttt{MARKETVAL} hidden.}
\label{tab:ahs-hidden-main}
\begin{tabular}{r cccc ccccc}
\toprule
 & \multicolumn{4}{c}{Fixed (ex-ante)} & \multicolumn{5}{c}{Contextual} \\
 \cmidrule(lr){2-5} \cmidrule(lr){6-10}
$k$ & Top-$k$ & Marg.\ val.\ & Fwd step.\ & Exact & Deviation & Marginal & Greedy & Smart & Exact \\
\midrule
0  & \multicolumn{4}{c}{1.0000} & \multicolumn{5}{c}{1.0000}  \\
1  & 0.8596 & 0.7616 & 0.7616 & 0.7616 & 0.9291 & 0.3917 & 0.3917 & 0.3917 & 0.3917 \\
2  & 0.6717 & 0.7169 & 0.6484 & 0.6484 & 0.6929 & 0.3479 & 0.2920 & 0.2920 & 0.2793 \\
3  & 0.6025 & 0.6908 & 0.5754 & 0.5754 & 0.6317 & 0.3423 & 0.2494 & 0.2494 & 0.2278 \\
4  & 0.5356 & 0.6822 & 0.5170 &        & 0.5881 & 0.3384 & 0.2233 & 0.2233 &  \\
5  & 0.5232 & 0.6312 & 0.4729 &        & 0.5481 & 0.3333 & 0.2048 & 0.2048 &  \\
8  & 0.4637 & 0.5047 & 0.4119 &        & 0.4927 & 0.3124 & 0.1725 & 0.1725 &  \\
10 & 0.4401 & 0.4296 & 0.3851 &        & 0.4596 & 0.3020 & 0.1603 & 0.1603 &  \\
\midrule
\multicolumn{3}{l}{All (43)} & 0.2363 &        & 0.2363 &  0.2363   & 0.2363 & 0.1322\rlap{$^\ddagger$} &        \\
\bottomrule
\end{tabular}

\vspace{0.5em}
\begin{minipage}{0.94\linewidth}
\footnotesize
\emph{Notes.} Lower values are better. The loss is the globally normalized weighted L2 recovery criterion in \eqref{eq:ahs-loss}, so the no-reveal benchmark equals one by construction. The sample contains 1,521 housing units in the Minneapolis--St.~Paul--Bloomington metro area with observed \texttt{MARKETVAL}. The state and loss use $\log(\texttt{MARKETVAL})$; dollar-scale price errors are reported only as diagnostics. The value coordinate is not directly revealable. ``Marg.\ val.''\ ranks features by the loss reduction from revealing each feature alone (covariance-aware but ignores redundancy). ``Fwd step.''\ is greedy forward selection conditioning on features already chosen. Fixed smart greedy is omitted because it coincides numerically with forward stepwise at the reported budgets. ``Deviation'' under contextual reveals features with the largest realized absolute deviations from their marginal means. ``Marginal'' under contextual denotes contextual marginal top-$k$ (covariance-ignorant but posterior-aware for singletons). ``Smart'' is contextual greedy with per-house early stopping. Fixed ``Exact'' and contextual ``Exact'' enumerate all subsets for $k\le3$ on the same full sample.
\\
${}^\ddagger$Smart greedy with budget $k=43$ stops at a median of 25 features per house.
\end{minipage}
\end{table}

\paragraph{Fixed versus contextual.}
Contextual policies uniformly outperform their fixed-policy counterparts at every budget we consider, and the gap is quantitatively large. At $k=1$, for example, fixed forward stepwise attains a loss of 0.7616, while contextual greedy attains 0.3917. 
The gap widens as the budget increases: at $k=10$, contextual greedy attains 0.1603 against 0.3851 for fixed forward stepwise. The advantage of contextual over fixed highlighting is therefore first-order in this realistic data environment.

\paragraph{The role of correlation in fixed policies.}
Among fixed policies, the comparison between marginal value and forward stepwise isolates the cost of ignoring correlation. 
Both policies use the full covariance structure to evaluate each feature's informativeness, but marginal value evaluates features one at a time while forward stepwise conditions on the features already selected. At $k=1$ the two coincide, both selecting \texttt{UNITSIZE}, the most informative single feature. From $k=2$ onward, however, marginal value stacks highly correlated size and room-count variables (\texttt{UNITSIZE}, \texttt{BATHROOMS}, \texttt{TOTROOMS}, \texttt{BEDROOMS}, \texttt{UNITFLOORS}), while forward stepwise diversifies quickly across distinct informational dimensions (\texttt{UNITSIZE}, \texttt{BLD}, \texttt{OWNLOT}). 
This pattern demonstrates that accounting for covariance in the evaluation of each individual feature is not sufficient for good fixed-policy performance; one must also account for redundancy among the features jointly selected.

\paragraph{Covariance awareness on both margins.}
Among fixed policies, forward stepwise uniformly improves on the covariance-ignorant top-$k$ ranking. The top-$k$ rule begins with highly weighted variables such as building type, bathrooms, and unit size, whereas the covariance-aware ordering diversifies more quickly. Similarly, among contextual policies, the simple deviation heuristic performs worst because it uses only first moments: its loss is 0.9291 at $k=1$ and 0.6317 at $k=3$, far above even contextual marginal top-$k$. The covariance-ignorant but posterior-aware marginal top-$k$ improves substantially on deviation, and contextual greedy performs best among the scalable contextual rules.

\paragraph{Non-monotonicity and the full-reveal benchmark.}
The bottom row of \autoref{tab:ahs-hidden-main} reports the loss when all 43 non-price features are revealed. This full-reveal benchmark attains a loss of 0.2363, which is the irreducible error from the \texttt{MARKETVAL} component that cannot be perfectly predicted from the other features. Notably, contextual greedy at $k=10$ already achieves 0.1603, substantially \emph{below} full reveal. Smart greedy with a budget of $k=43$ achieves 0.1322 by stopping at a median of 25 features per house, while contextual greedy without early stopping at $k=43$ matches the full-reveal loss of 0.2363. The explanation is that revealing low-weight, weakly correlated features can degrade the Gaussian posterior estimate of \texttt{MARKETVAL} by injecting noise through the covariance update. Early stopping avoids this by withholding features whose marginal effect on the posterior is harmful. This is the empirical counterpart of the non-monotonicity result in \autoref{ex:correlated-nonmonotone}: even in a high-dimensional, realistic setting, revealing more information need not improve outcomes.

\paragraph{Exact small-$k$ benchmarks.}
The exact enumeration results confirm that the scalable heuristics are directionally accurate. On the fixed side, exact ex-ante enumeration coincides with fixed forward stepwise for $k \leq 3$, selecting \{\texttt{UNITSIZE}\} at $k=1$, \{\texttt{BLD}, \texttt{UNITSIZE}\} at $k=2$, and \{\texttt{BLD}, \texttt{UNITSIZE}, \texttt{OWNLOT}\} at $k=3$. On the contextual side, the exact oracle equals greedy at $k=1$ and improves on contextual greedy at $k=2$ and $k=3$: exact contextual loss is 0.2793 versus 0.2920 for contextual greedy at $k=2$, and 0.2278 versus 0.2494 at $k=3$. This suggests that the sequential greedy rule captures much of the value of contextual adaptation while remaining computationally tractable.

\paragraph{Sophisticated-agent performance and robustness.}
We also report a stylized robustness exercise for sophisticated updating at $k\leq 3$. The naive benchmark remains the analytical Gaussian posterior throughout. For fixed policies, sophisticated and naive losses coincide because the selected set is constant ex ante. For contextual policies, we simulate sophisticated belief updates based on a simulation from the prior. Specifically, we take 100{,}000 draws from the fitted Gaussian prior, with each revealable synthetic non-price coordinate snapped to the nearest observed AHS support code before policy evaluation and conditioning.
Within this sample, we group support draws by the realized selected set and the displayed revealed covariates, and take the empirical conditional mean within that support sample. This procedure presents a way of calculating posteriors based on feature values \emph{and} feature selection.

\autoref{tab:ahs-soph-real} reports the performance of the resulting choices on the actual AHS sample, while \autoref{tab:ahs-soph-same} and \autoref{tab:ahs-soph-independent} report two synthetic validations. Under the correct model and exact conditioning, a sophisticated agent should weakly outperform a naive agent. The synthetic benchmarks line up with that logic, while the real-data benchmark remains slightly worse, which we interpret as evidence that the empirical selection-aware posterior is sensitive to misspecification when the true housing distribution deviates from the fitted Gaussian approximation. Furthermore, overfitting may deteriorate the performance of the more complex updating rule.

\begin{table}[p]
\centering
\caption{Performance of the sophisticated agent relative to the naive agent across simulation settings. Lower values are better.}
\label{tab:ahs-soph}

\begin{subtable}{\textwidth}
\centering
\caption{Real AHS sample with \texttt{MARKETVAL} hidden.}
\label{tab:ahs-soph-real}
\resizebox{\textwidth}{!}{%
\begin{tabular}{l cccc ccccc cccc}
\toprule
 & \multicolumn{4}{c}{Fixed (same for naive and sophisticated)} & \multicolumn{5}{c}{Contextual naive} & \multicolumn{4}{c}{Contextual sophisticated} \\
 \cmidrule(lr){2-5} \cmidrule(lr){6-10} \cmidrule(lr){11-14}
$k$ & Top-$k$ & Marg.\ val.\ & Fwd step.\ & Exact & Deviation & Marginal & Greedy & Smart & Exact & Deviation & Marginal & Greedy & Smart \\
\midrule
0 & 1.0000 & 1.0000 & 1.0000 & 1.0000 & 1.0000 & 1.0000 & 1.0000 & 1.0000 & 1.0000 & 1.0000 & 1.0000 & 1.0000 & 1.0000 \\
1 & 0.8596 & 0.7616 & 0.7616 & 0.7616 & 0.9291 & 0.3917 & 0.3917 & 0.3917 & 0.3917 & 0.9407 & 0.4393 & 0.4393 & 0.4393 \\
2 & 0.6717 & 0.7169 & 0.6484 & 0.6484 & 0.6929 & 0.3479 & 0.2920 & 0.2920 & 0.2793 & 0.7211 & 0.3665 & 0.2956 & 0.2956 \\
3 & 0.6025 & 0.6908 & 0.5754 & 0.5754 & 0.6317 & 0.3423 & 0.2494 & 0.2494 & 0.2278 & 0.7044 & 0.3637 & 0.2637 & 0.2637 \\
\bottomrule
\end{tabular}%
}
\begin{minipage}{0.96\linewidth}
\footnotesize
\emph{Notes.} Fixed columns are identical for naive and sophisticated agents. Contextual sophisticated columns use an empirical conditional mean built from 100{,}000 Gaussian support draws, conditioning on the realized selected set and the displayed revealed signal codes after snapping synthetic support draws to the empirical AHS signal alphabet.
\end{minipage}
\end{subtable}

\vspace{1em}

\begin{subtable}{\textwidth}
\centering
\caption{Same 100{,}000 Gaussian draws used for support and evaluation.}
\label{tab:ahs-soph-same}
\resizebox{\textwidth}{!}{%
\begin{tabular}{l cccc ccccc cccc}
\toprule
 & \multicolumn{4}{c}{Fixed (same for naive and sophisticated)} & \multicolumn{5}{c}{Contextual naive} & \multicolumn{4}{c}{Contextual sophisticated} \\
 \cmidrule(lr){2-5} \cmidrule(lr){6-10} \cmidrule(lr){11-14}
$k$ & Top-$k$ & Marg.\ val.\ & Fwd step.\ & Exact & Deviation & Marginal & Greedy & Smart & Exact & Deviation & Marginal & Greedy & Smart \\
\midrule
0 & 0.9023 & 0.9023 & 0.9023 & 0.9023 & 0.9023 & 0.9023 & 0.9023 & 0.9023 &  & 0.9023 & 0.9023 & 0.9023 & 0.9023 \\
1 & 0.7879 & 0.6862 & 0.6862 & 0.6862 & 0.8588 & 0.4221 & 0.4221 & 0.4221 &  & 0.8505 & 0.3407 & 0.3407 & 0.3407 \\
2 & 0.6182 & 0.6412 & 0.5920 & 0.5920 & 0.7194 & 0.3831 & 0.3233 & 0.3233 &  & 0.6751 & 0.2848 & 0.2140 & 0.2140 \\
3 & 0.5470 & 0.6166 & 0.5936 & 0.5936 & 0.6296 & 0.3723 & 0.2750 & 0.2750 &  & 0.4654 & 0.2336 & 0.1070 & 0.1070 \\
\bottomrule
\end{tabular}%
}
\begin{minipage}{0.96\linewidth}
\footnotesize
\emph{Notes.} Support and evaluation samples coincide, but the revealed synthetic signals are snapped to the empirical AHS support values, so this benchmark is no longer degenerate. Exact contextual columns are omitted because the exact contextual oracle is not recomputed on the synthetic samples.
\end{minipage}
\end{subtable}

\vspace{1em}

\begin{subtable}{\textwidth}
\centering
\caption{Independent 10{,}000-draw Gaussian evaluation sample.}
\label{tab:ahs-soph-independent}
\resizebox{\textwidth}{!}{%
\begin{tabular}{l cccc ccccc cccc}
\toprule
 & \multicolumn{4}{c}{Fixed (same for naive and sophisticated)} & \multicolumn{5}{c}{Contextual naive} & \multicolumn{4}{c}{Contextual sophisticated} \\
 \cmidrule(lr){2-5} \cmidrule(lr){6-10} \cmidrule(lr){11-14}
$k$ & Top-$k$ & Marg.\ val.\ & Fwd step.\ & Exact & Deviation & Marginal & Greedy & Smart & Exact & Deviation & Marginal & Greedy & Smart \\
\midrule
0 & 0.8899 & 0.8899 & 0.8899 & 0.8899 & 0.8899 & 0.8899 & 0.8899 & 0.8899 &  & 0.8899 & 0.8899 & 0.8899 & 0.8899 \\
1 & 0.7773 & 0.6739 & 0.6739 & 0.6739 & 0.8462 & 0.4170 & 0.4170 & 0.4170 &  & 0.8397 & 0.3389 & 0.3389 & 0.3389 \\
2 & 0.6142 & 0.6332 & 0.5788 & 0.5788 & 0.7116 & 0.3785 & 0.3191 & 0.3191 &  & 0.6985 & 0.2971 & 0.2382 & 0.2382 \\
3 & 0.5382 & 0.6092 & 0.5804 & 0.5804 & 0.6252 & 0.3681 & 0.2710 & 0.2710 &  & 0.6684 & 0.2992 & 0.2293 & 0.2293 \\
\bottomrule
\end{tabular}%
}
\begin{minipage}{0.96\linewidth}
\footnotesize
\emph{Notes.} The support sample remains the same 100{,}000 discretized Gaussian draws, but evaluation is on an independent synthetic sample. The sophisticated gains shrink relative to \autoref{tab:ahs-soph-same} because the empirical conditioning partition thins out of sample.
\end{minipage}
\end{subtable}

\end{table}

These tables support a narrow but useful conclusion. Once the sophisticated benchmark is put on the same discrete signal alphabet as the AHS data, the synthetic exercises line up with the theory: sophisticated updating improves on naive updating for the contextual policies, with the strongest gains when support and evaluation are drawn from the same discretized Gaussian environment. The deviation heuristic benefits from sophisticated updating as well in the same-sample synthetic benchmark, but its gains are much smaller than those of the posterior-based contextual rules and they disappear out of sample at $k=3$. On the real AHS sample, the empirical sophisticated proxy remains slightly worse than the Gaussian naive benchmark for all contextual policies, including deviation. We read that gap as evidence of misspecification of the Gaussian state model rather than as a substantive challenge to the theoretical dominance logic.
However, it points to the practical challenge that results are heavily driven by the exact way beliefs are updated.

\section{Extensions}\label{sec:extensions}

We discuss several extensions of our framework.
First, we comment on continuous instance spaces, where the set of possible states is infinite.
Second, we discuss alternative behavioral updating rules that interpolate between naive and sophisticated belief formation.
Third, we incorporate private human information, which affects belief formation.
Finally, we mention wrong belief and misaligned preferences as important extensions for future research.

\subsection{Continuous Features and Gaussian Highlighting}\label{sec:continuous}

Our main results focus on discrete distributions, where a highlighting policy maps each state in a finite support to a subset of coordinates.
In many applications, however, $X$ is continuous, and the posteriors induced by a policy can be even more complex.
In such settings, the designer's problem for a sophisticated receiver becomes a continuous-state information design problem.
A simple illustration is a Gaussian recovery problem.
Let $X=(X_1,X_2)\sim \mathcal N(0,I_2)$ and suppose the target is $y(X)=X\in\R^2$ with squared loss.
Let $k=1$ and restrict attention to policies that reveal exactly one coordinate value.
For a \emph{naive} agent, revealing coordinate $j$ sets $\hat X_j=X_j$ and leaves the other coordinate at its prior mean $0$.
The optimal naive policy is therefore to reveal the coordinate with larger magnitude,
i.e., $\sigma^N(x_1,x_2)\in \argmax_{j\in\{1,2\}}|x_j|$, yielding the expected loss
\[
R^N(\sigma^N)
=\E\!\left[\min\{X_1^2,X_2^2\}\right]
\approx 0.36.
\]
For a \emph{sophisticated} agent, however, the choice of which coordinate is revealed can itself convey information about the other coordinate.
Any such policy can be represented by two prediction functions $h,g:\R\to\R$:
if the policy reveals $X_1$, the receiver sets $(\hat X_1,\hat X_2)=(X_1,h(X_1))$;
if it reveals $X_2$, it sets $(\hat X_1,\hat X_2)=(g(X_2),X_2)$.
Thus the designer's objective can be transformed to finding $(h,g)$ to minimize
\[
J(h,g)\;=\;\mathbb{E}\!\left[\min\Big\{(X_2-h(X_1))^2,\ (X_1-g(X_2))^2\Big\}\right],
\qquad (X_1,X_2)\sim\mathcal N(0,I_2).
\]
The induced highlighting rule is
\[
\sigma(x_1,x_2)\;=\;
\begin{cases}
\{1\}, & \text{if } (x_2-h(x_1))^2 \le (x_1-g(x_2))^2, \\
\{2\}, & \text{otherwise}.
\end{cases}
\]
Even in this two-dimensional case, optimizing over $(h,g)$ is nontrivial.
\autoref{fig:naive-vs-complex} contrasts the naive partition induced by $\sigma^N$ with a more intricate partition induced by a numerically optimized pair $(h,g)$ (computed via a Lloyd-style best-response iteration), which attains expected loss
\[
R^S(\sigma^{\text{complex}})\approx 0.28 < 0.36 = R^S(\sigma^N)\,
\]
a reduction of about $22\%$ relative to $\sigma^N$.
This example illustrates that in continuous settings, sophisticated receivers turn the highlighting problem into a continuous information-design problem, and optimal policies may exhibit complex state-dependent regions even with just two dimensions and a highly symmetrical prior.

\begin{figure}[h]
  \centering
  \includegraphics[width=0.8 \linewidth]{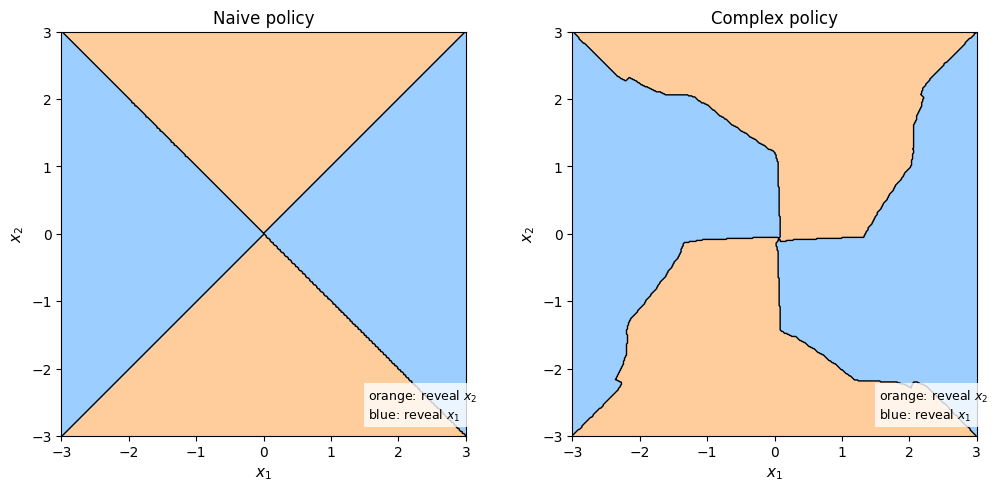}
  \caption{Naive vs.\ complex highlighting in the Gaussian recovery problem with $k=1$.
  The naive policy is optimal for an agent updating naively, while the complex policy is a numerical approximation of the optimal policy for a sophisticated agent.
  }
  \label{fig:naive-vs-complex}
\end{figure}

\subsection{Alternative Behavioral Updating and Stereotypes}

Our main model considers two benchmark belief-formation rules: sophisticated updating, which conditions on the selection event, and naive updating, which ignores it.
In practice, both assumptions are natural, yet extreme benchmarks.
Real decision-maker updating may also be influenced by anchoring, availability, and representativeness.
As a specific behavioral deviation, we consider \emph{stereotypical updating} \citep[e.g.][]{bordalo2016stereotypes}, where the agent replaces Bayesian conditioning by finding a representative instance.
One stylized version is the following: after observing $(I,X_I)$, the agent selects a single ``representative''
\[
\hat X \in \argmax_{x\in\supp(\P):\ \ x_I= X_I}\ \P(X = x \mid \ X_I),
\]
and then plays the action that is optimal for $\hat X$.
Under squared recovery loss, this corresponds to imputing all missing coordinates by the most likely completion rather than the posterior mean.
A detailed study of optimal highlighting against such behavioral rules is an interesting direction for future work.

\subsection{Private Human Information and a Minimax Property}
\label{sec:private-info}

Above, we have assumed that a decision instance is fully characterized by the vector $x$ that is available to the algorithm.
In practice, the human decision-maker may have additional private information.
We now allow the human to observe such private information $X_H$ in addition to the algorithm's message $(I,X_I)$.
The posterior definitions in \autoref{sec:model} extend directly:
\begin{align*}
\hat{\P}^{S}(\cdot\mid X_H,I,X_I)
&=\P(\cdot\mid X_H,X_I,\sigma(X_M)=I),\\
\hat{\P}^{N}(\cdot\mid X_H,I,X_I)
&=\P(\cdot\mid X_H,X_I).
\end{align*}
The design objective remains to minimize expected loss, except that the Bayes action is now computed under posteriors conditional on $X_H$.

Private information also provides an additional motivation for robust design.
If the algorithm is uncertain about what $X_H$ the human has access to, a natural robust objective is to choose a policy that performs well uniformly over possible private information structures.
Because additional private information can only reduce a Bayesian receiver's expected squared error, the worst case for the designer is the case in which the human observes nothing beyond the highlighted features.
Consequently, finding the policy that minimizes expected loss under the assumption that $X_H=\emptyset$ may itself be viewed as a robust solution to the problem with private information.

Formally, fix a jointly distributed model for $(X_M,X_H)$ with $X_M\in\R^d$ the machine features available to the designer and $X_H$ an additional square-integrable signal that the human may or may not observe.
Write $\mathcal H$ for the collection of such private information structures; the degenerate element $\emptyset\in\mathcal H$ corresponds to a human with no private information.
For a highlighting policy $\sigma$ with $|\sigma(x)|\le k$ and a receiver type $\tau\in\{S,N\}$, let $\hat X_M^\tau(\sigma;X_I,X_H)$ denote the induced Bayes action for reconstructing $X_M$ (note that $X_H$ is always perfectly available, so does not add loss), and define the risk with private information $H$ by
\[
R^\tau(\sigma;H)\ :=\ \E\!\left[\bigl\|X_M-\hat X_M^\tau(\sigma;X_I,X_H)\bigr\|_2^2\right],\qquad I=\sigma(X_M),
\]
with $R^\tau(\sigma;\emptyset)$ the risk defined in \autoref{sec:model}.
Throughout what follows, the revealed vector $X_I$ is understood to implicitly carry its index tag $I$, so that $\Sigma(X_I)=\Sigma(I,X_I)$ is a well-defined sub-$\sigma$-algebra of the underlying probability space.

\begin{proposition}[No private information policies are minimax for the sophisticated agent]\label{prop:minimax-private}
For any highlighting policy $\sigma$ and any private information structure $H\in\mathcal H$,
\[
R^S(\sigma;H)\ \le\ R^S(\sigma;\emptyset).
\]
Consequently, any policy that is optimal in the no private information problem,
\[
\sigma^\star\ \in\ \argmin_{\sigma:\,|\sigma(x)|\le k}\ R^S(\sigma;\emptyset),
\]
is also minimax-optimal over private information structures,
\[
\sigma^\star\ \in\ \argmin_{\sigma:\,|\sigma(x)|\le k}\ \sup_{H\in\mathcal H}\ R^S(\sigma;H).
\]
\end{proposition}

The above minimax proposition implies that optimizing the highlighting policy for a sophisticated human with no private information delivers a robust solution against arbitrary private information channels the designer may not observe or control.
However, the result does not carry over to the naive agent.
For the sophisticated agent, the posterior mean is the orthogonal projection of $X_M$ onto the full information $\sigma$-field generated by the private signal and the highlighting rule, so adding information can only reduce squared loss.
The naive agent ignores the informational content of the selection event, so its action is generally not the conditional expectation with respect to that richer $\sigma$-field.
As a result, the projection argument behind \autoref{prop:minimax-private} breaks down for naive updating.

\subsection{Wrong Beliefs and Misaligned Preferences}
\label{sec:inefficiencies}

The above model assumes that both sophisticated and naive agents have a prior of the world that coincides with the actual distribution.
Furthermore, principal (designer) and agent (decision-maker) choices are fully aligned.
An important natural extension is to a model where both priors and preferences can differ between principal and agent.
Specifically, the agent may start with a prior $\tilde{\P}$ and minimize expected loss for a loss function $\tilde{\ell}$. The principal may then still optimize for a highlighting policy for their prior $\P$ and loss function $\ell$, taking both frictions into account and connecting the model further to Bayesian persuasion \citep{kamenica2011bayesian}.
Highlighting may be an attractive way of collaboration in this case since merely revealing rather than manipulating features may still lead to reasonable solutions even in a world where principal and agent have partially misaligned preferences or different beliefs and may not fully trust each other.

\section{Conclusion}\label{sec:conclu}

We consider a model of highlighting in human--AI collaboration where an algorithm reveals a subset of features to enable a decision-maker to make better decisions about a given problem instance.
Our main takeaway is that the way in which the decision-maker interprets the provided information matters for how to design good algorithms.
When the decision-maker interprets the values of the revealed features in isolation, simple algorithms perform close to optimally.
However, when the decision-maker also interprets \emph{why} information was revealed, optimal algorithms are more complex and may confuse naive decision-makers.
Beyond these specific assumptions about how information is interpreted, our approach points to an important agenda of modeling how decision-makers interpret information and deriving tractable highlighting algorithms as a practical tool to achieve human--AI complementarity.

\bibliography{reference}

\appendix

\section*{Proofs}
\begin{proof}[Proof of \autoref{thm:nphard-sophisticated-opt}]
We prove NP-hardness by a polynomial-time reduction from the Euclidean sum-of-squares clustering
decision problem with $2$ clusters, which is NP-hard by \cite{aloise2009np}.

\paragraph{Source problem (Euclidean $2$-means).}
An instance consists of points $z_1,\dots,z_m\in\mathbb R^p$ and a threshold $T\in\mathbb R_{\ge 0}$.
The question is whether there exists a partition $(C_1,C_2)$ of $[m]$ such that
$$
\sum_{t=1}^2\ \sum_{i\in C_t}\|z_i-\mu_t\|_2^2 \ \le\ T,
\qquad\text{where}\qquad
\mu_t:=\frac{1}{|C_t|}\sum_{i\in C_t} z_i .
$$

\paragraph{Reduction outline.}
Given $(z_1,\dots,z_m,T)$, we construct in polynomial time a highlighting instance $(\mathbb P,A,\ell)$
(with $X_M\subseteq\{0,1\}^d$ and $k=1$) and a scalar $\widetilde T:=T/n$ such that
$$
\min_{\sigma} R^S(\sigma)\ \le\ \widetilde T
\quad\Longleftrightarrow\quad
\text{the $2$-means instance has optimum }\le T.
$$
Therefore, if we could compute $\min_\sigma R^S(\sigma)$, we could solve the NP-hard $2$-means decision
problem by comparing the returned optimum to $\widetilde T$. Hence computing $\min_\sigma R^S(\sigma)$
(or an optimizer) is NP-hard.

\paragraph{Construction.}
Fix a $2$-means instance $(z_1,\dots,z_m,T)$.

\textbf{Step 1: choose the feature dimension.}
Let
$$
d := \max\{6,\ 2+\lceil \log_2 m\rceil\},
$$
so that $2^{d-2}\ge m$.

\textbf{Step 2: define the state space and prior.}
Assume no private human information ($X_H=\emptyset$). The machine-observed instance is binary:
$X_M\in\{0,1\}^d$.

\emph{(i) Data states.}
Let $\mathrm{enc}:[m]\to\{0,1\}^{d-2}$ be any injective map. For each $i\in[m]$, define
$$
y^i := (0,0,\mathrm{enc}(i))\in\{0,1\}^d.
$$

\emph{(ii) Gadget states (one per non-cluster signal).}
Under bandwidth $k=1$, a deterministic policy can realize signals of the form $s=(I,x_I)$ where
$I\in\{\emptyset,\{1\},\dots,\{d\}\}$ and $x_I\in\{0,1\}^{|I|}$. Thus there are $2d+1$ possible signals:
$$
\mathcal S = \{\emptyset\}\ \cup\ \{(\{j\},0),(\{j\},1): j\in[d]\}.
$$
Reserve the two ``cluster signals'' $(\{1\},0)$ and $(\{2\},0)$ for data states, and let
$$
\mathcal S_0 := \mathcal S\setminus\{(\{1\},0),(\{2\},0)\},
\qquad\text{so}\qquad |\mathcal S_0|=2d-1.
$$
Index $\mathcal S_0$ as $\mathcal S_0=\{s^0,s^1,\dots,s^{2d-2}\}$.

For each $t\in\{0,1,\dots,2d-2\}$, define a distinct gadget state $x^t\in\{0,1\}^d$ satisfying:
\begin{enumerate}
\item $x^t_1=x^t_2=1$ (so $x^t$ can never realize $(\{1\},0)$ or $(\{2\},0)$);
\item $x^t$ can realize its designated signal $s^t$ (i.e., if $s^t=\emptyset$ choose $I=\emptyset$; if
$s^t=(\{j\},b)$ then ensure $x^t_j=b$ and choose $I=\{j\}$);
\item the vectors $x^t$ are all distinct.
\end{enumerate}
Such a collection exists for $d\ge 6$.

Let
$$
\mathrm{supp}(\mathbb P) := \{y^1,\dots,y^m\}\cup\{x^0,\dots,x^{2d-2}\},
\qquad n:=m+(2d-1),
$$
and let $\mathbb P$ be uniform over $\mathrm{supp}(\mathbb P)$.

\textbf{Step 3: define the action space and loss.}
Let
$$
A := \mathbb R^p\ \cup\ \{0,1,\dots,2d-2\},
\qquad
B := T+1,
\qquad
\widetilde T := \frac{T}{n}.
$$
Define $\ell:A\times\mathrm{supp}(\mathbb P)\to\mathbb R_{\ge 0}$ by:
\begin{enumerate}
\item For gadget states $x^t$,
$$
\ell(a,x^t)=
\begin{cases}
0, & \text{if } a=t,\\
B, & \text{otherwise}.
\end{cases}
$$
\item For data states $y^i$,
$$
\ell(a,y^i)=
\begin{cases}
\|a-z_i\|_2^2, & \text{if } a\in\mathbb R^p,\\
B, & \text{if } a\in\{0,1,\dots,2d-2\}.
\end{cases}
$$
\end{enumerate}

\paragraph{Completeness.}
Assume the $2$-means instance is a YES instance: there exists a partition $(C_1,C_2)$ with cost $\le T$.
Define a policy $\sigma$ (with $k=1$) by:
\begin{enumerate}
\item For each gadget state $x^t$, choose an index set realizing its designated signal $s^t\in\mathcal S_0$.
\item For each data state $y^i$, set
$$
\sigma(y^i)=
\begin{cases}
\{1\}, & \text{if } i\in C_1,\\
\{2\}, & \text{if } i\in C_2.
\end{cases}
$$
Since $y^i_1=y^i_2=0$, data states realize exactly $(\{1\},0)$ and $(\{2\},0)$.
\end{enumerate}
Then each gadget signal $s^t$ is produced only by $x^t$, so the sophisticated posterior is degenerate
and the Bayes action is $a=t$, yielding zero gadget loss. For the two data signals, the posterior is supported
on $\{y^i:i\in C_1\}$ and $\{y^i:i\in C_2\}$; under squared loss, the Bayes action is the conditional mean,
i.e., the corresponding centroid $\mu_1,\mu_2$. Hence
$$
R^S(\sigma)
=\frac{1}{n}\sum_{i=1}^m \|z_i-\mu_{\mathrm{cl}(i)}\|_2^2
\le \frac{T}{n}
=\widetilde T,
$$
so $\min_\sigma R^S(\sigma)\le \widetilde T$.

\paragraph{Soundness.}
Conversely, suppose there exists a policy $\sigma$ with $R^S(\sigma)\le \widetilde T=T/n$.

\textbf{Claim 1: no signal pools two distinct gadget states.}
If some realized signal $s$ is generated by distinct gadget states $x^t$ and $x^{t'}$ with $t\neq t'$,
then any single action $a$ incurs loss at least $B$ on at least one of $\{x^t,x^{t'}\}$, implying
overall expected loss at least $B/n > T/n$, a contradiction.

\textbf{Claim 2: no signal pools a gadget state with a data state.}
If some signal $s$ is generated by $x^t$ and $y^i$, then any action $a\in A$ incurs loss at least $B$ on at least
one of them (vector actions incur $B$ on gadgets; discrete actions incur $B$ on data), giving expected loss
at least $B/n>T/n$, a contradiction.

\textbf{Claim 3: gadget states occupy all signals in $\mathcal S_0$.}
Each gadget state has $x^t_1=x^t_2=1$, so it cannot realize $(\{1\},0)$ or $(\{2\},0)$ under $k=1$; hence it must
realize some signal in $\mathcal S_0$. By Claim 1, the $2d-1$ gadget states realize $2d-1$ distinct signals in
$\mathcal S_0$, which has size $2d-1$, so they realize all of $\mathcal S_0$.

\textbf{Data states can use only two signals.}
By Claim 2, data states cannot realize any signal in $\mathcal S_0$. Therefore every data state must realize one
of the remaining two signals $(\{1\},0)$ or $(\{2\},0)$, inducing a partition $(C_1,C_2)$ of $[m]$.
Under squared loss, the sophisticated Bayes action for each signal is the centroid of the corresponding cluster,
so the data-state contribution equals
$$
\frac{1}{n}\sum_{t=1}^2\ \sum_{i\in C_t}\|z_i-\mu_t\|_2^2.
$$
Since $R^S(\sigma)\le T/n$, multiplying by $n$ yields
$$
\sum_{t=1}^2\ \sum_{i\in C_t}\|z_i-\mu_t\|_2^2\ \le\ T,
$$
so the original $2$-means instance is a YES instance.

\paragraph{Polynomial-time reduction.}
The construction is polynomial in the input size: $d=O(\log m)$ and $n=m+O(\log m)$, and the instance is binary with
$k=1$. Therefore deciding whether $\min_\sigma R^S(\sigma)\le \widetilde T$ is NP-hard. As explained above, this implies
computing $\min_\sigma R^S(\sigma)$ (or an optimizer) is NP-hard.
\end{proof}

\begin{proof}[Proof of \autoref{lm:fixed-indep}]
If a fixed set $I$ is revealed, the naive Bayes action is $a_i=X_i$ for $i\in I$ and $a_i=p_i$ for $i\notin I$.
Thus $\mathrm{loss}(I;X)=\sum_{i\notin I}(X_i-p_i)^2$, so
\[
R^N(I)=\sum_{i\notin I}\E[(X_i-p_i)^2]=\sum_{i\notin I}\Var(X_i)=\sum_{i\notin I}p_i(1-p_i).
\]
Minimizing over $|I|\le k$ is equivalent to maximizing $\sum_{i\in I}p_i(1-p_i)$, hence choose the $k$ largest $p_i(1-p_i)$.
\end{proof}

\begin{proof}[Proof of \autoref{lm:contextual-indep}]
Fix $x$, for any $I$, $L^N(I;x)=\sum_{i\notin I}(x_i-p_i)^2 = \sum_{i=1}^d (x_i-p_i)^2-\sum_{i\in I}(x_i-p_i)^2$.
The first term is constant in $I$, so minimizing loss is equivalent to maximizing $\sum_{i\in I}(x_i-p_i)^2$ over $|I|\le k$, i.e., reveal the $k$ largest surprises.
\end{proof}

\begin{proof}[Proof of \autoref{prop:fixedasymptotic}]

Write $F^*_d, Q^*_d$ for the cdf and quantile functions for a uniform draw from the $(p^*_{d,i})_{i=1}^d$.
By \autoref{asm:stability}, $F^*_d(t) \rightarrow F^*(t)$ for all continuity points of $F^*$.
Furthermore, the same holds for the quantile functions $Q^*_d$ and their limit $Q^*$.

Now let $U \sim \text{Uniform}([0,1])$.
Then $Q^*_d(U), Q^*(U)$ are distributed according to $F^*_d, F^*$, respectively.
In particular, for the normalized risk of a fixed subset faced by either the sophisticated or naive agent, we have that, as $d \rightarrow \infty$,
\begin{align*}
    &\E\left[
        \frac{1}{d} \sum_{j=1}^{d - k}
        (X^*_{d,(j)} - p^*_{d,(j)})^2
    \right]
    =
    \frac{1}{d} \sum_{j=1}^d \Ind(j \leq d - k)
        \: p^*_{d,(j)} (1 - p^*_{d,(j)})
    \\
    &=
    \E[
        \Ind(U \leq \frac{d-k}{d})
        Q^*_d(U) (1 - Q^*_d(U))
    ]
    =
    \E[
        \Ind(U \leq 1 - \alpha)
        Q^*_d(U) (1 - Q^*_d(U))
    ]
    + o(1)
    \\
    &\rightarrow
    \E[
        \Ind(U \leq 1 - \alpha)
        Q^*(U) (1 - Q^*(U))
    ]
    =
    \int_{q = 0}^{1-\alpha}
    Q^*(q) (1 - Q^*(q)) \de q.
    \qedhere
\end{align*}
\end{proof}

\begin{proof}[Proof of \autoref{prop:fractionasymptotic}]

We use the same argument and notation as in the proof of \autoref{prop:fixedasymptotic}.

\paragraph{Sophisticated agent.}
Under fixed surprise highlighting, the sophisticated agent learns the value of the first $\lfloor \beta d \rfloor$ features, but not anything about the remaining ones. Hence, the expected normalized loss equals
\begin{align*}
    &\E\left[
        \frac{1}{d} \sum_{j=\lfloor \beta d \rfloor + 1}^{d}
        (X^*_{d,(j)} - p^*_{d,(j)})^2
    \right]
    =
    \frac{1}{d} \sum_{j=1}^{d} \Ind(j > \lfloor \beta d \rfloor)\: p^*_{d,(j)} (1 - p^*_{d,(j)})
    \\
    &=
    \E\left[\Ind\!\left(U > \frac{\lfloor \beta d \rfloor}{d}\right)\, Q^*_d(U)\bigl(1-Q^*_d(U)\bigr)\right]
    =
    \E\left[\Ind(U > \beta)\, Q^*_d(U)\bigl(1-Q^*_d(U)\bigr)\right] + o(1)
    \\
    &\rightarrow
    \E\left[\Ind(U > \beta)\, Q^*(U)\bigl(1-Q^*(U)\bigr)\right]
    =
    \int_{q=\beta}^{1} Q^*(q)\bigl(1-Q^*(q)\bigr)\,\de q.
\end{align*}

\paragraph{Naive agent.}
For $j\le \lfloor \beta d\rfloor$, the feature is highlighted iff $X^*_{d,(j)}=1$, so the naive agent suffers loss on coordinate $j$ iff $X^*_{d,(j)}=0$, in which case the squared error equals $(0-p^*_{d,(j)})^2=(p^*_{d,(j)})^2$.
Thus, the naive agent has \emph{additional} expected loss
\begin{align*}
    &\E\left[
        \frac{1}{d}
        \sum_{j= 1}^{\lfloor \beta d \rfloor}
        \Ind(X^*_{d,(j)} {=} 0)
        (X^*_{d,(j)} - p^*_{d,(j)})^2
    \right]
    =
    \frac{1}{d}\sum_{j=1}^{d}\Ind(j\le \lfloor \beta d\rfloor)\: (p^*_{d,(j)})^2\,(1-p^*_{d,(j)})
    \\
    &=
    \E\left[\Ind\!\left(U \le \frac{\lfloor \beta d \rfloor}{d}\right)\, \bigl(Q^*_d(U)\bigr)^2\bigl(1-Q^*_d(U)\bigr)\right]
    =
    \E\left[\Ind(U \le \beta)\, \bigl(Q^*_d(U)\bigr)^2\bigl(1-Q^*_d(U)\bigr)\right] + o(1)
    \\
    &\rightarrow
    \E\left[\Ind(U \le \beta)\, \bigl(Q^*(U)\bigr)^2\bigl(1-Q^*(U)\bigr)\right]
    =
    \int_{q=0}^{\beta} \bigl(Q^*(q)\bigr)^2\bigl(1-Q^*(q)\bigr)\,\de q.
\end{align*}
Moreover, for $j\ge \lfloor \beta d\rfloor + 1$ the procedure never highlights, so the expected loss on those coordinates is the same as for the sophisticated agent computed above.
Combining the two parts yields
\begin{align*}
    &\E\left[
        \frac{1}{d}
        \sum_{j= 1}^{\lfloor \beta d \rfloor}
        \Ind(X^*_{d,(j)} {=} 0)
        (X^*_{d,(j)} - p^*_{d,(j)})^2
        +
        \frac{1}{d}
        \sum_{j=\lfloor \beta d \rfloor + 1}^{d}
        (X^*_{d,(j)} - p^*_{d,(j)})^2
    \right]
    \\
    &\rightarrow
    \int_{q=0}^{\beta} \bigl(Q^*(q)\bigr)^2\bigl(1-Q^*(q)\bigr)\,\de q
    +
    \int_{q=\beta}^{1} Q^*(q)\bigl(1-Q^*(q)\bigr)\,\de q.
    \qedhere
\end{align*}
\end{proof}

\begin{proof}[Proof of \autoref{prop:equivalence}]
Fix \(\beta\in(0,1)\).
Define the (random) fraction of features revealed by the fixed surprise-highlighting \autoref{enum:fraction}, which are the ones among the first \(\lfloor \beta d\rfloor\) coordinates, as
\[
\hat{\alpha}_d(\beta)
:=
\frac{1}{d}\sum_{j=1}^{\lfloor \beta d\rfloor} X^*_{d,(j)}.
\]
Given the (non-random) vector \((p^*_{d,(1)},\ldots,p^*_{d,(d)})\), the variables \(X^*_{d,(j)}\) are independent Bernoulli with means \(p^*_{d,(j)}\).
Hence
\[
\E\bigl[\hat{\alpha}_d(\beta)\bigr]
=
\frac{1}{d}\sum_{j=1}^{\lfloor \beta d\rfloor} p^*_{d,(j)},
\qquad
\Var\bigl(\hat{\alpha}_d(\beta)\bigr)
=
\frac{1}{d^2}\sum_{j=1}^{\lfloor \beta d\rfloor} p^*_{d,(j)}\bigl(1-p^*_{d,(j)}\bigr)
\le
\frac{\lfloor \beta d\rfloor}{4d^2}
\le
\frac{1}{4d}.
\]
Furthermore, as in the proof of \autoref{prop:fixedasymptotic},
\begin{align*}
    &\E\bigl[\hat{\alpha}_d(\beta)\bigr]
    =
    \frac{1}{d}\sum_{j=1}^{d} \Ind(j\le \lfloor \beta d\rfloor)\: p^*_{d,(j)}
    \\
    &=
    \E\Bigl[\Ind\!\left(U\le \frac{\lfloor \beta d\rfloor}{d}\right)\,Q^*_d(U)\Bigr]
    =
    \E\bigl[\Ind(U\le \beta)\,Q^*_d(U)\bigr] + o(1)
    \\
    &\rightarrow
    \E\bigl[\Ind(U\le \beta)\,Q^*(U)\bigr]
    =
    \int_{0}^{\beta} Q^*(q)\,\de q
    =:
    \overline{\alpha}(\beta).
\end{align*}
As a consequence,
\begin{equation}
\label{eqn:alphaconvergence}
    \hat{\alpha}_d(\beta) \xrightarrow{p} \overline{\alpha}(\beta)
\end{equation}
as \(d \rightarrow \infty\).

The function \(\overline{\alpha}(\beta)\) is continuous in \(\beta\).
By \autoref{asm:smallbandwidth}, there exists a \(\beta^* \in (0,1)\) such that \(\overline{\alpha}(\beta^*) = \alpha\), and \(\overline{\alpha}(\beta)\) is strictly increasing in a neighborhood of \(\beta^*\).

Define now \(\hat{\beta}_d\) as the (random) normalized highest index revealed by the greedy rule~\autoref{enum:greedy},
\[
\hat{\beta}_d
=
\frac{1}{d}
\min\left(\left\{j \in \{1,\ldots,d\}: \sum_{j'=1}^j X^*_{d,(j')} = k \right\} \cup \{d\}\right),
\]
where we choose \(d\) if all ones are revealed.
Then the greedy algorithm is equivalent to the algorithm that highlights all ones among a fraction \(\hat{\beta}_d\) of cases.

For a sufficiently small \(\varepsilon > 0\), consider now fixed-highlighting algorithms with fractions \(\beta^* - \varepsilon\), \(\beta^* + \varepsilon\).
By \autoref{eqn:alphaconvergence}, we have that
\begin{align*}
    \hat{\alpha}_d(\beta^* - \varepsilon) &\xrightarrow{p} \overline{\alpha}(\beta^* - \varepsilon),
    &
    \hat{\alpha}_d(\beta^* + \varepsilon) &\xrightarrow{p} \overline{\alpha}(\beta^* + \varepsilon).
\end{align*}
Since also \(\overline{\alpha}(\beta^* - \varepsilon) < \alpha < \overline{\alpha}(\beta^* + \varepsilon)\) for \(\varepsilon\) small enough,
\[
\P\bigl(\hat{\alpha}(\beta^* - \varepsilon) \le \alpha \le \hat{\alpha}(\beta^* + \varepsilon)\bigr) \to 1.
\]
Since \(\beta^* - \varepsilon \le \hat{\beta}_d \le \beta^* + \varepsilon\) if and only if \(\hat{\alpha}_d(\beta^* - \varepsilon) \le \alpha \le \hat{\alpha}_d(\beta^* + \varepsilon)\),
\[
\P\bigl(\beta^* - \varepsilon \le \hat{\beta}_d \le \beta^* + \varepsilon\bigr) \to 1.
\]
Hence,
\begin{equation}
    \hat{\beta}_d \xrightarrow{p} \beta^*.
\end{equation}
As a direct consequence, normalized expected risks are equivalent.
\end{proof}

\begin{proof}[Proof of \autoref{thm:greedyasymptotic}]
    The result follows from combining \autoref{prop:equivalence} with \autoref{prop:fractionasymptotic} (evaluated at \(\beta=\beta^*\)).
\end{proof}

\begin{proof}[Proof of \autoref{prop:dev-weak-corr-general}]
Fix an instance $x$ and a revealed set $I$ with $|I|\le k$.
Under squared loss, the naive Bayes action predicts each unrevealed coordinate $i\notin I$ by
$m_i(I,x_I):=\E[X_i\mid X_I=x_I]$.
Define the realized loss
$$
L^N(I;x):=\sum_{i\notin I}(x_i-m_i(I,x_I))^2
\qquad\text{and the ``surrogate''}\qquad
S(I;x):=\sum_{i\notin I}(x_i-\mu_i)^2.
$$
By \autoref{asm:weak-corr-mean}, for each $i\notin I$ we have $|m_i(I,x_I)-\mu_i|\le\varepsilon$.
Let $\delta_i:=m_i(I,x_I)-\mu_i$ so $|\delta_i|\le\varepsilon$.
Since $x_i\in[a_i,b_i]$ and $\mu_i\in[a_i,b_i]$, we have $|x_i-\mu_i|\le R$.
Expanding,
$$
(x_i-m_i)^2=(x_i-\mu_i-\delta_i)^2=(x_i-\mu_i)^2-2\delta_i(x_i-\mu_i)+\delta_i^2,
$$
so
$$
\big|(x_i-m_i)^2-(x_i-\mu_i)^2\big|\le 2|\delta_i|\,|x_i-\mu_i|+\delta_i^2\le 2R\varepsilon+\varepsilon^2.
$$
Summing over at most $d-k$ unrevealed coordinates gives a uniform perturbation bound:
$$
|L^N(I;x)-S(I;x)|\le (d-k)(2R\varepsilon+\varepsilon^2)=:C.
$$
\autoref{alg:dev} minimizes $S(I;x)$ over $|I|\le k$ for each $x$.
Using $L^N\le S+C$ and $S\le L^N+C$ to compare the minimizers of $S$ and $L^N$ yields
$L^N(I^{\mathrm{dev}}(x);x)\le L^N(I^\star(x);x)+2C$ pointwise.
Taking expectation over $X$ proves the proposition.
\end{proof}
\begin{proof}[Proof of \autoref{prop:minimax-private}]
We first establish the following fact: for any square-integrable random vector $Y$ and any two $\sigma$-algebras with $\mathcal G\subseteq\mathcal G'$,
\begin{equation}\label{eq:l2-proj}
\E\bigl\|Y-\E[Y\mid \mathcal G']\bigr\|_2^2\ \le\ \E\bigl\|Y-\E[Y\mid \mathcal G]\bigr\|_2^2,
\end{equation}
with equality iff $\E[Y\mid\mathcal G']=\E[Y\mid\mathcal G]$ almost surely.
Since $\mathcal G\subseteq\mathcal G'$, the tower property gives $\E[Y\mid\mathcal G]=\E\!\left[\E[Y\mid\mathcal G']\,\big|\,\mathcal G\right]$, so we can decompose the coarser error as
\[
Y-\E[Y\mid\mathcal G]\ =\ \underbrace{\bigl(Y-\E[Y\mid\mathcal G']\bigr)}_{=:A}\ +\ \underbrace{\bigl(\E[Y\mid\mathcal G']-\E[Y\mid\mathcal G]\bigr)}_{=:B}.
\]
The two terms are orthogonal in $L^2$: $B$ is $\mathcal G'$-measurable, and $\E[A\mid\mathcal G']=\E[Y\mid\mathcal G']-\E[Y\mid\mathcal G']=0$, so $\E\!\left[A^\top B\right]=\E\!\left[\E[A\mid\mathcal G']^\top B\right]=0$.
Pythagoras' identity then yields
\[
\E\bigl\|Y-\E[Y\mid\mathcal G]\bigr\|_2^2\ =\ \E\bigl\|Y-\E[Y\mid\mathcal G']\bigr\|_2^2\ +\ \E\bigl\|\E[Y\mid\mathcal G']-\E[Y\mid\mathcal G]\bigr\|_2^2.
\]
The last term is non-negative, which gives \eqref{eq:l2-proj}, and it vanishes iff $\E[Y\mid\mathcal G']=\E[Y\mid\mathcal G]$ almost surely, which gives the equality condition.

For a fixed highlighting policy, write $I=\sigma(X_M)$ for the selected index set; both $I$ and the (index-tagged) revealed coordinates $X_I$ are well-defined random variables, and we write $\Sigma(Z)$ for the $\sigma$-algebra generated by a random variable $Z$ (distinct from the highlighting map $\sigma$).
Because $X_I$ carries its index tag, $\Sigma(X_I)=\Sigma(I,X_I)$, and the sophisticated agent's information under each private information structure is
\[
\mathcal G^S_{\emptyset}\ :=\ \Sigma(I,X_I)\ =\ \Sigma(X_I),
\qquad
\mathcal G^S_{H}\ :=\ \Sigma(I,X_I,X_H)\ =\ \Sigma(X_I,X_H).
\]
These satisfy the nesting $\mathcal G^S_{\emptyset}\subseteq\mathcal G^S_{H}$ because adjoining the additional generator $X_H$ can only enlarge the generated $\sigma$-algebra.
Because the sophisticated posterior is $\hat{\P}^S(\cdot\mid X_H,I,X_I)=\P(\cdot\mid X_H,X_I,\sigma(X_M)=I)$, the sophisticated Bayes action for $X_M$ under squared loss is the corresponding conditional expectation. 

Applying \eqref{eq:l2-proj} with $Y=X_M$, $\mathcal G=\mathcal G^S_{\emptyset}$, and $\mathcal G'=\mathcal G^S_{H}$ yields
\[
R^S(\sigma;H)\ =\ \E\bigl\|X_M-\E[X_M\mid\mathcal G^S_{H}]\bigr\|_2^2\ \le\ \E\bigl\|X_M-\E[X_M\mid\mathcal G^S_{\emptyset}]\bigr\|_2^2\ =\ R^S(\sigma;\emptyset).
\]
Equality holds when $\E[X_M\mid\mathcal G^S_H]=\E[X_M\mid\mathcal G^S_\emptyset]$ a.s., e.g.\ when $X_H$ is conditionally independent of $X_M$ given $\mathcal G^S_\emptyset=\Sigma(I,X_I)$.
The upper bound $R^S(\sigma;\emptyset)$ does not depend on $H$, so $\sup_{H\in\mathcal H}R^S(\sigma;H)=R^S(\sigma;\emptyset)$ for every highlighting policy, and the outer minimizations coincide.
\end{proof}
\end{document}